%% file: arxiv.tex
\documentclass{article}
\usepackage[hyphens]{url}  
\usepackage{graphicx} 
\usepackage{natbib}  
\usepackage{caption} 
\usepackage{subcaption}
\usepackage[most]{tcolorbox}
\tcbset{colback=gray!10!white, colframe=black, boxrule=0.5pt, arc=2pt, left=6pt, right=6pt, top=4pt, bottom=4pt}
\usepackage{amsmath, amssymb, amsthm}  
\newtheoremstyle{theoremstyle}
  {10pt}   
  {10pt}   
  {\itshape}  
  {}       
  {\bfseries}  
  {.}      
  { }      
  {}       
\newtheoremstyle{definitionstyle}
  {10pt}
  {10pt}
  {\normalfont}
  {}
  {\bfseries}
  {.}
  { }
  {}
\theoremstyle{theoremstyle}
\newtheorem{theorem}{Theorem}
\newtheorem{lemma}{Lemma}

\theoremstyle{definitionstyle}
\newtheorem{definition}{Definition}
\newtheorem{algorithm}{Algorithm}
\usepackage{xparse} 
\usepackage{mathtools} 
\usepackage{graphicx} 
\usepackage{todonotes}
\usepackage{xspace}
\usepackage{cancel}
\usepackage{algorithm}
\usepackage{algpseudocode}

\usepackage{xcolor}
\definecolor{DeepBlue}{rgb}{0.1, 0.2, 0.6}
\definecolor{DeepGreen}{rgb}{0.0, 0.4, 0.2}
\definecolor{WarmGray}{rgb}{0.4, 0.4, 0.4}

\usepackage{silence}
\NewDocumentCommand{\set}{s m}{\IfBooleanTF{#1}{{\{#2\}}}{{\ensuremath{\left\{#2\right\}}}}}
\NewDocumentCommand{\abs}{s m}{\IfBooleanTF{#1}{{|#2|}}{{\ensuremath{\left|#2\right|}}}}
\NewDocumentCommand{\norm}{s m}{\IfBooleanTF{#1}{{\ensuremath{\left\|#2\right\|}}}{{\|#2\|}}}

\NewDocumentCommand{\IBU}{}{\mathtt{IBU}}

\NewDocumentCommand{\DKL}{}{{\mathrm{D}_{\mathrm{KL}}}}

\NewDocumentCommand{\inner}{s m}{\IfBooleanTF{#1}{{\langle#2\rangle}}{{\ensuremath{\left\langle#2\right\rangle}}}}

\NewDocumentCommand{\argmin}{}{\operatorname*{arg\,min}}

\NewDocumentCommand{\setargmin}{}{\operatorname*{arg\,min}^{\text{(set)}}}
\NewDocumentCommand{\setargmax}{}{\operatorname*{arg\,max}^{\text{(set)}}}

\def\oo{\infty}
\def\eps{\epsilon}

\def\M{\mathcal M\xspace}
\def\niters{{N_\text{iters}}}
\def\Var{\text{Var}}
\def\setMLE{\mathbf{MLE}}
\def\th{\theta}
\def\ph{\phi}
\def\pq{(p{-}q)}
\def\K{K}
\def\bbR{{\mathbb R}}
\def\Kout{{\K_{\text{out}}}}
\def\Kin{{\K_{\text{in}}}}
\def\INV{\texttt{Inv}\xspace}
\def\INVN{{\texttt{InvN}}\xspace}
\def\INVP{{\texttt{InvP}}\xspace}
\def\IBU{{\texttt{IBU}}\xspace}
\def\ourMLE{{\texttt{MLE*}}\xspace}

\def\nth{^{[n]}}
\def\nthprev{^{[n-1]}}
\def\nthnext{^{[n+1]}}

\def\sign{\textnormal{sign}}

\makeatletter\newcommand{\manualref}[2]{\protected@write\@auxout{}{\string\newlabel{#1}{{#2}{999}}}}\makeatother
\manualref{thm:new-main}{7}
\manualref{thm:lagrange}{4}
\manualref{thm:monotonicity}{6}
\manualref{lemma:props-nth}{3}
\manualref{lemma:n-smallest}{4}
\manualref{thm:collinearity}{8}
\manualref{thm:consistency}{9}
\manualref{app:results}{D}

\author{
\begin{tabular}[t]{c}
Carlos Antonio Pinzón\textsuperscript{\rm 1},
Ehab ElSalamouny\textsuperscript{\rm 1,\rm 5},
Lucas Massot\textsuperscript{\rm 3},\\
Alexis Miller\textsuperscript{\rm 4},
Héber Hwang Arcolezi\textsuperscript{\rm 2}, 
Catuscia Palamidessi\textsuperscript{\rm 1}\end{tabular}\\
{ }\\
\textsuperscript{\rm 1}{\textnormal{INRIA Saclay, France}}\\
\textsuperscript{\rm 2}{\textnormal{INRIA Grenoble, France}}\\
\textsuperscript{\rm 3}{\textnormal{École Polytechnique, France}}\\
\textsuperscript{\rm 4}{\textnormal{Ecole Normale Supérieure de Lyon, France}}\\
\textsuperscript{\rm 5}{\textnormal{Suez Canal University, Egypt}}\\
}
\title{Estimating the True Distribution of Data Collected with Randomized Response}

\begin{document}
\bibliographystyle{alpha}

\maketitle

\begin{abstract}
    Randomized Response (RR) is a protocol designed to collect and analyze categorical data with local differential privacy guarantees.
    It has been used as a building block of mechanisms deployed by Big tech companies to collect app or web users' data.
    Each user reports an automatic random alteration of their true value to the analytics server, which then estimates the histogram of the true unseen values of all users using a debiasing rule to compensate for the added randomness.
    A known issue is that the standard debiasing rule can yield a vector with negative values (which can not be interpreted as a histogram), and there is no consensus on the best fix.
    An elegant but slow solution is the Iterative Bayesian Update algorithm (IBU), which converges to the Maximum Likelihood Estimate (MLE) as the number of iterations goes to infinity.
    This paper bypasses IBU by providing a simple formula for the exact MLE of RR and compares it with other estimation methods experimentally to help practitioners decide which one to use.
\end{abstract}

\section{Introduction}

Local differential privacy (LDP)~\citep{duchi2013local} is a framework for providing privacy guarantees when collecting data from a set of users.
It removes the need for trust in the correct management of whoever collects the data and is typically used for automated telemetry.
The essence of LDP is to introduce controlled uncertainty right before the users transmit their data to avoid the data collector from receiving the exact specific information of each user.

The level of uncertainty or noise introduced by LDP is controlled by a parameter $\epsilon>0$, varying from a high privacy regime with $\epsilon\approx0$ (large uncertainty) to a low privacy regime with $\epsilon\gg 1$ (low uncertainty). 
The organization that collects and processes the data can not be certain about the specifics of any user in particular, but it can still derive statistical estimations.
These estimations become more precise as $\epsilon$ increases and as the size of the user group grows.

One of the simplest and most fundamental LDP mechanisms is the Randomized Response   ~\cite{kairouz2016discrete} (also called $\K$-ary RR, kRR, or simply RR), introduced firstly without any connection to LDP  in ~\cite{warner1965randomized}.
In a nutshell, RR reports either the truth with some fixed probability or a random value, picked uniformly.

Despite its simplicity, RR has been the subject of extensive research as it is the building block of more advanced mechanisms like Symmetric Unary Encoding and Local Hashing~\citep{wang2017locally} as well as longitudinal protocols, e.g.  RAPPOR~\citep{erlingsson2014rappor}, d-bitFlipPM~\citep{ding2017collecting}, Longitudinal Local Hashing~\citep{arcolezi2022frequency}.

Moreover, RR is known to be optimal among all LDP mechanisms for small domains or large values of $\epsilon$. Namely, whenever the domain size $\K$ satisfies $\K < 3 e^\epsilon + 2$~\cite{wang2017locally}.
Large values of $\epsilon$ are particularly relevant in light of recent advances in Shuffle differential privacy~\cite{erlingsson2019amplification,cheu2019distributed}, where a shuffler is applied after local randomization to strip metadata and randomly permute the reports. 
This added layer of anonymity enables privacy amplification, meaning that stronger central-DP guarantees can be achieved from weaker local ones.  The growing use of Shuffle DP in real-world scenarios heightens the importance of high $\epsilon$ regimes, and therefore, of RR.

Concerning statistical estimations, we focus on estimating the original data distribution, which is arguably the most fundamental statistic. To this purpose, typically the analytics server constructs a histogram of reported data, normalizes it to obtain a distribution, and uses a debiasing linear correction rule (\INV) to compensate for the added noise.
The resulting vector is proven to be unbiased, but it may not be a valid distribution, as it may have some negative entries.
\begin{table}
    \centering
    \renewcommand{\arraystretch}{1.5}
    \small
    \begin{tabular}{|l|l|l|}
        \hline
        \textbf{Method} & \textbf{Complexity} & \textbf{Justification} \\
        \hline
        Debiasing \INV   & $O(\K)$ & \checkmark{} Unbiased, but invalid \\
        Neg. fix \INVN   & $O(\K)$ & Workaround: simple \\
        Proj. fix \INVP  & $O(\K\log\K)$ & Workaround: close to \INV \\
        Bayesian \IBU & $O(\K\,\niters)$ & \checkmark{} MLE when $\niters\to\oo$\\
        Proposed \ourMLE    & $O(\K\log\K)$ & \checkmark{} MLE in one step\\
        \hline
    \end{tabular}
    \caption{Methods for estimating the original distribution from RR observations.  $\niters$ represents the number of iterations, which depends on the desired precision. }
    \label{tab:methods}
\end{table}
In order to obtain a distribution, two simple workarounds have been proposed. The first,  \INVN, sets to $0$ all the negative entries of the vectors, and renormalizes it. The second, \INVP,  projects the vector on the closest point on the simplex in terms of Euclidean distance. 

A different approach, called iterative Bayesian update (IBU), was proposed by ~\citep{agrawal2001design}. IBU  produces an estimate of the original distribution 
 using an iterative algorithm that converges to the maximum likelihood estimate (MLE) as the number of iterations grows~\cite{elsalamouny2020generalized}. 
In this sense, IBU is well justified mathematically. However, the asymptotic nature of the algorithm makes it less attractive. Another drawback is the missing theory about the number of iterations or the stop condition that guarantees a certain proximity to the  MLE.

Our first objective is to enhance the efficiency of MLE computation. To this aim, we propose $\ourMLE$, based on a mathematical formula that can be computed quickly, and which provides the exact MLE, thus bypassing the issue about the stop condition. 

The second objective is to compare these estimators. One may think that, since the  MLE is the ``most likely'' estimate, it would also be the most precise. Furthermore, the MLE is known to minimize the Kullback-Leibler divergence between the empirical distribution (i.e., the normalized histogram output by RR) and the ideal distribution, obtained by multiplying the input distribution by the channel matrix that represents the RR noise.  However, for practical purposes, the Kullback-Leibler divergence may not be the most important utility metric. Indeed,  the precision of an estimation is typically measured in terms of mean square error (MSE). In this paper, we compare the MSE precision of the MLE,  \INVN and \INVP. It turns out that, surprisingly,  the precision of \INVN and \INVP ``flip'' depending on the distribution being more or less concentrated, while that of the MLE is always in between.  

Other important terms of comparisons we investigate in this paper are efficiency,  consistency, and unbiasedness. The methods that we consider are summarized in Table~\ref{tab:methods}.

In summary, the  contributions of this paper are:
\begin{enumerate}
    \item 
We provide a simple formulation for the MLE, prove its correctness formally, and validate it empirically. These formulas coincide with \cite{kairouz2016discrete}, Supplementary Material, Section F. 
    \item
We propose an algorithm that computes the MLE and is significantly more efficient than the state-of-the-art solution (IBU).
In addition, our algorithm gives an exact solution, whereas IBU provides only an approximation because the number of iterations is necessarily finite.
    \item
We theoretically and experimentally compare MLE with other estimation methods to clarify their trade-offs, with the objective to help developers choose which one to use.
\end{enumerate}

\section{Preliminaries}

This section formulates the estimation problem, presents the state-of-the-art estimators from related work.

\subsection{Problem formulation}

There are $N$ users $u=1..N$, each of which has a secret value $x_u$ from a set of $\K$ categories labeled as $\{1,2,\ldots, \K\}$.
The secret value can represent, for example, the user's browser homepage, the number of times they click on a specific button, or the emoji they use most frequently.

An organization, referred to as the data collector or the analyst, is interested in finding out the most frequent categories of the population as a whole (not user by user), and more generally, they are interested in estimating the proportion $\theta_i$ of users $u$ taking value $x_u = i$ for every $i\in\set{1..\K}$.
This corresponds to finding a vector $\theta \in \Delta$ where the simplex $\Delta$ is the space of distributions over $\{1..\K\}$, i.e., $\Delta := \set{\th\in\bbR^\K: \sum_i \th_i = 1, \th_i\geq 0}$.

To avoid violating the privacy of the users, the data collector will not collect $x_u$ directly.
Instead, they will use a \emph{mechanism} $\M$ (a function whose output is influenced by randomness, also known as a channel in information theory), whose purpose is to deliberately transform the input data $x_u$ into some $y_u := \M(x_u)$ with controlled randomness to hide its true value while still providing some information about it.
One sample per user is measured.

We suppose that $\M$ is an instance of RR, and address the problem of how the analyst efficiently 
estimates the unknown distribution $\th\in\Delta$ using the outputs $(y_u)_{u=1}^N$ and 
the structure of the mechanism $\M$. 

\subsection{LDP and the RR mechanism}

Local Differential Privacy (LDP)~\citep{duchi2013local} is a strong privacy protection 
ensured by the user-side mechanism by setting a formal bound on how much information is 
revealed about the true input.
%
Formally, a mechanism $\M$ with discrete output space satisfies $\epsilon$-LDP for some $\epsilon>0$ (the smaller the more private) if for all inputs $x,x'$ and outputs $y$, it holds that 
$$\Pr(\M(x)=y) \leq e^\epsilon \Pr(\M(x')=y).$$
This protects the unknown input because upon observing any output $y$, the degree to which any candidate input $x$ is more likely than any other $x'$ is limited by $\epsilon$.

The randomized response (RR) mechanism takes an input $x\in\set{1..\K}$
and reports $\M(x)$ from the same domain such that
\begin{equation}\label{eq:rr}
    \Pr(\M(x)=y)=\begin{cases}
        p &\text{ if }y=x\\
        q = \frac{1-p}{\K-1} &\text{ otherwise}
    \end{cases}
\end{equation}
That is, $\M(x)$ returns $x$ with probablility $p$, and otherwise outputs some $x'\sim \text{uniform}(\set{1,...,\K}\setminus\set{x})$. Equivalently, it returns $x$ with probability $p-q$, 
and otherwise returns some $x'\sim \text{uniform}(\set{1,...,\K})$.
%
The RR was first introduced by~\citet{warner1965randomized} for the binary case ($\K=2$), as a survey technique for eliminating evasive answer bias.
More recently, RR has gained increased attention in the context of LDP, since it 
satisfies $\epsilon$-LDP for $\epsilon = \log(p/q)$.

\subsection{Estimators based on expectation}

Let $\ph\in\Delta$ denote the normalized histogram of the observed 
$y_u$ for $u=1..\K$,  after applying RR parametrized by some known value $p$.
Note that the expected value of $\ph$ is $\ph_i = q + (p-q) \th_i$.
Based on this observation,~\citet{kairouz2016discrete} derived a simple and unbiased estimator, which we call the \emph{linear inversion} estimator (\INV)
\begin{equation}
    \hat\th_i^\INV := \INV(\ph)_i := \frac{\ph_i-q}{p-q}.
    \label{eq:INV}
\end{equation}
Nevertheless, the constraint $\hat\th^\INV\in\Delta$ may fail to hold, 
since \eqref{eq:INV}  may produce negative values. 
%
%
This occurs very often, as detailed later in Section~\ref{sec:theoretical:validity}.
To overcome this issue, the following two post-processing solutions have been proposed in the literature.
\begin{enumerate}
    \item Normalization (\INVN): $\hat\th^\INVN_i \propto \max(0, \hat\th^\INV_i)$, i.e., set to zero the negative components and rescale.
    \item Projection (\INVP): $\hat\th^\INVP:= \argmin_{\th} \norm{\th - \hat\th^\INV}$, where $\th$ varies in the set of valid distributions.
\end{enumerate}
These estimators are valid (they always produce distributions), and they are fast.
\INVN is (clearly) $O(\K)$ and \INVP can be implemented in $O(\K\log\K)$.
However, their design is a post-processing workaround without formal guarantees.

\section{Estimators Based on Likelihood}


Let $\setMLE(\ph)$ be the set of all MLEs. That is
\begin{equation}
    \setMLE(\ph) := \setargmax_{\th\in\Delta} \Pr(\ph | \th),
\end{equation}
where $\Pr(\ph|\th)$ denotes the probability that the normalized histogram of the 
random variables $Y_u$ for $u=1..N$, is $\ph$ provided 
that the distribution for $x_u$ is $\th$ and $Y_u\sim \M(x_u)$.
Based on the theory of the Expectation-Maximization algorithm~\citep{dempster1977maximum}, 
it has been shown~\citep{agrawal2001design,agrawal2005privacy,elsalamouny2020generalized} 
that one can asymptotically approach an MLE using an algorithm known as 
Iterative Bayesian Update (IBU). This algorithm is applicable to any discrete mechanism 
with a finite channel matrix $C_{i\,j}:=\Pr(\M(i)=j)$ of size $\Kin \times \Kout$.

IBU starts with a fully supported prior distribution $\hat \th^{(0)}$, by default $\hat \th^{(0)}:=(1/\Kin,1/\Kin,...,1/\Kin)$, and repeatedly updates $\hat \th^{(t)}$ into $\hat \th^{(t+1)}$ using what is known as Jeffrey's update rule~\citep{jacobs2021learning,pinzón2025jeffreysupdateruleminimizer}:
\begin{equation}
    \hat \th^{(t+1)}_i := \sum_{j=1}^\Kout \frac{\ph_j\; \hat \th^{(t)}_i C_{i\,j}}{\sum_{k=1}^\Kin \hat \th^{(t)}_k C_{k\,j}}.
    \label{eq:IBU}
\end{equation}
The most important property of IBU is that $\hat \th^{(t)}$ converges to $\hat \th^*$ for some $\hat \th^* \in \setMLE(\ph)$ as $t\to\oo$~\cite{elsalamouny2020generalized}.
The complexity of this procedure for a fixed large number of iterations $t=\niters$ is $O(\Kin\,\Kout\,\niters)$, because at each time step $t$, all denominators (for every $j$) can be cached in $O(\Kin\,\Kout)$ and then $\hat \th^{(t+1)}$ can be computed in $O(\Kin\,\Kout)$.

In the particular case of RR, where $\Kin=\Kout=K$, the algorithm can be accelerated from $O(\K^2\,\niters)$ to $O(\K\,\niters)$ using symmetries. Precisely, the update step \eqref{eq:IBU} can be simplified to the following, 
which takes only $O(\K)$. 
\begin{equation}
    \begin{aligned}
    s^{(t+1)} &:= \sum_{i=1}^\K \frac{\ph_i}{q + \pq\, \hat\th^{(t)}_i}
    \\
    \hat\th^{(t+1)}_i &:= \hat\th^{(t)}_i \left(q\, s^{(t+1)} + \frac{\pq\, \ph_i}{q + \pq\,\hat\th^{(t)}_i}\right).
    \end{aligned}
    \label{eq:IBU-RR}
\end{equation}
We will denote this procedure by $\IBU(\ph):=\hat \th^{(\niters)}$ for some fixed $\niters$.



Lastly,~\citet{ye2025revisiting} propose an algorithm that optimizes a regularized version of the likelihood by merging small values together.
Their algorithm runs in $O(\K^2 \log(\K)\,\niters)$ and uses a smoothing factor, which, for some practical experiments, gives better results than a pure MLE.
\citet{hay2009boosting} address the estimation under central (not local) differential privacy, and \citet{lee2015maximum} propose an approximate iterative method (like IBU).

\section{Formula for the MLE}

This section starts with a derivation of the formula for the MLE, which coincides in different notation with that of \cite{kairouz2016discrete}, Supplementary Material, Section F, and then provides the pseudocode to compute it and a summarizing sketch of the detailed self-contained proof in the Supplementary Material.

Consider a RR mechanism, defined with $p>0$ and $q=\frac{1-p}{\K-1}$ as in \eqref{eq:rr}.
For any observed distribution $\ph$ and some threshold $\tau \in[\min_i \ph_i, \max_i \ph_i]$, define $\ph^\tau$ as
\[
\ph^\tau_i = \begin{cases}
    q & \text{if } \ph_i < \tau\\
    c_\tau\, \ph_i & \text{otherwise,}
\end{cases}
\qquad c_\tau=\frac{1-\sum_{\ph_i<\tau} q}{\sum_{\ph_i \geq \tau} \ph_i}.
\]
Notice that $\sum_i \ph^\tau_i = 1$ by the definition of the constant $c_\tau$. 
Observe also that the denominator in $c_\tau$ is non-zero because 
$\tau \leq \max_i \ph_i$.
For this threshold transformation, notice that
\begin{equation}
    \INV(\ph^\tau)_i = \begin{cases}
        0 & \text{if } \ph_i < \tau\\
        \frac{c_\tau \ph_i\ - q}{p - q} & \text{otherwise.}
    \end{cases}
    \label{eq:tau-estimator}
\end{equation}
The vector $\INV(\ph^\tau)$ sums up to $1$, but it may contain negative values for small $\tau$.
The proposed estimator, which we call $\ourMLE(\ph)$, is precisely $\INV(\ph^{\tau^*})$ where 
$\tau^*$ is the smallest threshold $\tau$ for which $\INV(\ph^\tau)$ does not contain negative values:
\begin{equation}
\begin{aligned}
    &\hat\th^\ourMLE := \ourMLE(\ph) := \INV(\ph^{\tau^*}),\\
    &\tau^* := \min \set{\tau\; \middle|\; \forall i,\,\ph_i {<} \tau \vee \;c_\tau\,\ph_i {\geq} q }.
\end{aligned}
    \label{eq:tau-star} 
\end{equation}

\def\idx#1{{\sigma({#1})}}

Algorithm~\ref{alg:proposed} computes $\hat\th^\ourMLE$ for 1-indexed arrays, and we give 
a zero-indexed Python implementation in the Supplementary Material.
The following invariant holds during the loop: $s$ is the suffix sum $s = \sum_{i>k} \ph_{\idx{i}}$. 
After the loop, $i$ is the number of zeros in the output $\hat\th^\ourMLE$ and $\tau^* = \phi_{\sigma(i)}$.
The main loop that keeps track of $s$ and finds the threshold $\tau^*$ is $O(\K)$, 
therefore, the algorithm as a whole is $O(\K\log \K)$ because it sorts the indices of $\ph$ 
before entering the main loop.

\begin{algorithm}
\caption{\label{alg:proposed} Proposed algorithm \ourMLE. Indexed from 1. }
\begin{algorithmic}[1]
\State \textbf{Input:} $\ph_{1..\K}$, $p$, $q$.
\State $\sigma \gets$ argsort($\ph_{1..\K}$) \Comment{$\ph_{\idx{1}}\leq \cdots\leq \ph_{\idx{\K}}$}
\State $i \gets 0,\quad s \gets 1$ 
\While{$i < \K$ and $q s > \ph_{\idx{1+i}}\,(1 - i q)$}
    \State $s \gets s - \ph_{\idx{1+i}},\quad i \gets i + 1$
\EndWhile
\State $\hat\th^*_{\idx{j}} \gets 0$ for each $j=1,...,i$.
\State $\hat\th^*_{\idx{j}} \gets \frac{\ph_{\idx{j}}\, (1 - i q)  - s q}{s \pq}$  for each $j=i{+}1, ..., \K$.
\State \Return $\hat\th^*_{1..\K}$
\end{algorithmic}
\end{algorithm}

We conclude this section by showing that \ourMLE is the unique MLE of the RR mechanism, and we show a theoretical application of this result.

\begin{theorem}\label{thm:main-text}
    The MLE for the RR mechanism is unique and given by the proposed formula \ourMLE.
\end{theorem}
\newenvironment{proofsketch}
  {\noindent\textit{Proof sketch.} }
  {\hfill$\square$\par}
\begin{proofsketch}
The complete proof is given in Theorem~\ref{thm:new-main} (Supplementary Material). The following is a sketch of it.

Suppose that $\hat\th$ is an MLE.
It can be shown using the Lagrange Multipliers method that for every category $i$, either $\hat\th_i=0$ or $\hat\th_i = \frac{\ph_i}{\lambda} - \frac{q}{p-q}$, where the value of $\lambda$ is such that the total sum is 1, i.e. $\lambda := \frac{(p-q) \sum_{\hat\th_i>0} \ph_i}{1- |\set{i:\hat\th_i{=}0}| q}$.
This is proven in Theorem~\ref{thm:lagrange} (Supplementary Material).
With this result, the problem of computing all the components of $\hat\th$ is reduced into identifying the components $i$ in which $\hat\th_i=0$ because $\lambda$ and all the other components can be computed with the given formulas.
The rest of the proof is devoted to identifying the zero-valued components.

As shown in Theorem~\ref{thm:monotonicity} (Supplementary Material), the entries in $\hat\th$ are 
monotonic with respect to the entries in $\ph$ in the sense that $\hat\th_i\leq\hat\th_j$ if and only 
if $\ph_i\leq\ph_j$. This observation follows by contradiction: if $\hat\th$ was not monotonic with 
respect to $\ph$, then the order of two entries can be flipped, resulting in a new $\hat\th'$ that 
has a higher likelihood.  

Therefore, the zeros of $\hat\th$ must occur in the positions of the smallest $n$ components of $\ph$.
In other words, letting $\th\nth$ be the result of applying the aforementioned formula obtained via Lagrange Multipliers assuming that the zeros occur in the positions of the smallest $n$ values of $\ph$, then the MLE must be $\th\nth$ for some $n\in\set{0,...,\K-1}$ whose value is unknown so far.

Finally, to find the value of $n$, let $e_i$ denote the value of the $i$'th smallest component of $\phi$ and define $g(n):= (1-nq)e_{n+1} - q (\sum_{i>n} e_i)$ for $n<\K$, then the desired $n$ must satisfy $g(n)\geq 0$ and must be as small as possible.
This is shown in Lemmas~\ref{lemma:props-nth} and~\ref{lemma:n-smallest} (Supplementary material).
The threshold $\tau^*$ in Equation~\eqref{eq:tau-star} corresponds to $e_{n+1}$.

In summary, $\hat\th=\th\nth$, so the MLE is unique.
\end{proofsketch}

Thanks to Theorem~\ref{thm:main-text}, it is possible to study and prove analytical properties of the MLE, like the following, which partly explains some of the results in Section~\ref{sub:exp_main_results}.

\begin{theorem}
    If the smallest entry in $\phi$, say $\phi_i$, takes value $\ph_i<q$, and the second smallest is at least $(\K q - \ph_i) / (\K - 1)$, then the points in $\bbR^\K$ given by $\hat\th^\INVP$, $\hat\th^\ourMLE$ and $\hat\th^\INVN$ are collinear, with $\hat\th^\ourMLE$ in the middle.
\end{theorem}
\begin{proof}
    Proven in Theorem~\ref{thm:collinearity} (Supplementary Material).
\end{proof}


\section{Theoretical Comparison}

This section compares validity, unbiasedness, MSE, consistency and complexity for the estimators of interest: \ourMLE, \INV, \INVN and \INVP.

\subsection{Validity}\label{sec:theoretical:validity}

By validity, we denote the property that the returned estimates are always guaranteed to be valid distributions.
\INV guarantees $\sum_i \hat \th_i=1$ but not $\hat \th_i\geq 0$, therefore, it is invalid.
The other three estimators are valid by design.

The problem that \INV produces vectors with negative values occurs very frequently in practice.
Indeed, for a population with $\th_i\approx 0$ for some $i$, the probability of $\hat \th_i< 0$ equals that of $\ph_i< q$, which is mostly governed by the probability of a Binomial($N-1, q$) not exceeding $Nq$.
This value is close to $1/2$ and therefore non-negligible regardless of $N$.

\subsection{Unbiasedness}

It was shown in~\citet{wang2017locally}, Theorems 1 and 2, that \INV is unbiased and has element-wise variance given by
\begin{equation}
\Var(\hat \th_i) = \frac{q(1-q)}{N\pq} + \th_i \frac{\pq(1-2q-\pq)}{N\pq},
    \label{eq:var-inv}
\end{equation}
which follows from the observation that $\ph_i$ follows a Multinomial distribution.

\INVN, \INVP and \ourMLE, however, are biased for certain $\th$, and this is an inevitable consequence of being valid.
For instance, fix $\th=(1, 0, ..., 0)$ and vary $\ph$.
The realizations of the random vectors $\ph$ and $\hat \th$ are points scattered around $\ph'=(\pq, q, ..., q)$ and $\th$ respectively.
Some of the points around $\th$ are outside the valid region and some are inside, but they balance and we have $E(\hat \th)=\th$.
If we apply any rule $f$ that moves invalid points to the valid region and keeps valid points as they are, then $E(f(\hat \th))$ will necessarily fall in the (strict) interior of the valid region.
Thus, since $\th$ is in the border, not in the interior, we obtain $E(f(\hat \th))\ne \th$.
A generic proof for this fact can be found in~\citet{berger1990inadmissibility}.

\subsection{Complexity}

In terms of computational speed, $\INV$, $\INVN$ and $\INVP$ are $O(\K)$ while $\IBU$ is $O(\K\,\niters)$ and $\ourMLE$ is  $O(\K\log \K)$ (cost of sorting the indices).
The larger cost of $\ourMLE$ above $O(\K)$ is very low relative to the gained guarantees: the output is valid, it produces the exact MLE, and the complexity difference is sub-polynomial.

\subsection{MSE, TV and consistency}

\def\MSE{\text{MSE}}
The Mean Squared Error is defined as $\MSE_\th(\hat \th):=E_{\ph|\th}(\norm{\hat \th-\th}^2)$, where $\hat \th$ is fixed to one of the four estimators.
As shown in Theorem~\ref{thm:consistency} (Supplementary Material), the four estimators are consistent because $\MSE_\th(\hat \th) \in O(\K/N)$, which converges to $0$ as $N\to\oo$.
This implies that other measures, like the Total Variation $\text{TV}_\th(\hat \th):=\frac{1}{2} E_{\ph|\th}(\norm{\hat \th-\th}_1)$ also converge to $0$ as $N\to\oo$.
Alternatively, the consistency of $\ourMLE$ can be proven using the sufficient conditions derived by \citet{elsalamouny2025consistencyperformanceiterativebayesian} for the consistency of MLE for generic privacy mechanisms, of which RR is a particular case.

\section{Experiments} \label{sec:experiments}

In this section, we describe our experimental setup, present key results, and discuss the implications of our findings. The code is available at \cite{ourRepository}.  
Our experiments pursue two main objectives:  
\textbf{(1)} validating the correctness of our \ourMLE{} estimator by comparing it against the \IBU{} approach under the RR mechanism, and  \textbf{(2)} evaluating the performance and robustness of \ourMLE{} compared to two standard baselines, \INVP{} and \INVN.    
Specifically, while \INVP{} and \INVN{} perform well in distribution-specific settings, \ourMLE{} offers a consistently strong performance across a wide range of scenarios, making it a robust and reliable choice when the true data distribution is unknown.

\subsection{General Setup} \label{sub:setup_exp}

We conduct an extensive empirical evaluation to compare the performance of three estimators under the RR mechanism: 
\ourMLE, \INVP, and \INVN.

Experiments are designed to assess robustness across a wide range of configurations. 
Specifically, we vary:
\begin{itemize}
    \item Privacy budget $\epsilon \in \{1, 2, \ldots, 9, 10\}$;

    \item Sample size $N \in \{10^2, 10^3, 10^4, 10^5, 10^6\}$;

    \item Domain size $K \in \{50, 100, 1000, 5000\}$;

    \item Data skewness, controlled via a Zipf parameter $s$:
    \begin{itemize}
        \item Low concentration (near-uniform): $s=0.01$, 
        \item Moderate concentration: $s=1.3$, 
        \item High concentration: $s=2.5$. 
    \end{itemize}
\end{itemize}

Each configuration is repeated with 100 different random seeds to account for statistical variability, and the performance is evaluated using: (i) Mean Squared Error (MSE), and (ii) Negative Log-Likelihood.

\subsection{Overview of Results} \label{sub:exp_main_results}

To support a comprehensive evaluation, we conduct an extensive set of experiments varying the key parameters of Section~\ref{sub:setup_exp}. 
All experimental results, covering the full grid of configurations, are provided in the supplementary material (Section~\ref{app:results}). 
In this section, we highlight representative scenarios that illustrate the main trends and insights, focusing on how each estimator performs under different data regimes and parameter settings.

\paragraph{Evaluation.} Figure~\ref{fig:overview_results_fixed_eps} presents performance results for a practically relevant industrial setting with a large fixed domain size $K = 10{,}000$ and privacy level $\epsilon = 4.0$, while varying the number of users $N$. 
This choice of $K$ reflects the high-cardinality domains commonly encountered in real-world applications, such as telemetry data by Google Chrome and Microsoft Windows~\cite{erlingsson2014rappor,ding2017collecting}. 
We evaluate estimator performance across three representative data distributions, characterized by Zipf concentration parameters: \textbf{$s = 0.01$ (near-uniform)}, \textbf{$s = 1.3$ (moderately skewed)}, and \textbf{$s = 2.5$ (highly skewed)}.
Figure~\ref{fig:overview_results_fixed_n} complements this view by fixing the number of users to $N = 10^6$ and instead varying the privacy budget $\epsilon$, providing insights into estimator behavior across different privacy regimes.




\begin{figure}[!htb]
    \centering

    \begin{subfigure}[b]{1\linewidth}
        \includegraphics[width=\linewidth]{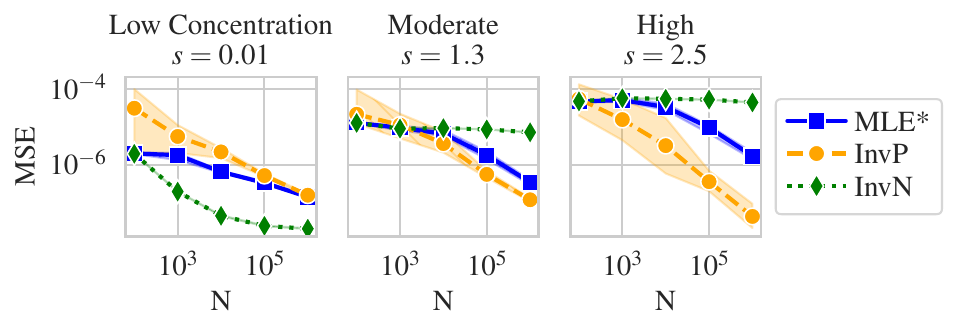}
        \caption{Mean Squared Error (MSE), $K=10000, \eps=4$}
        \label{fig:overview_mse_fix_eps}
    \end{subfigure}
    \hfill
    \begin{subfigure}[b]{1\linewidth}
        \includegraphics[width=\linewidth]{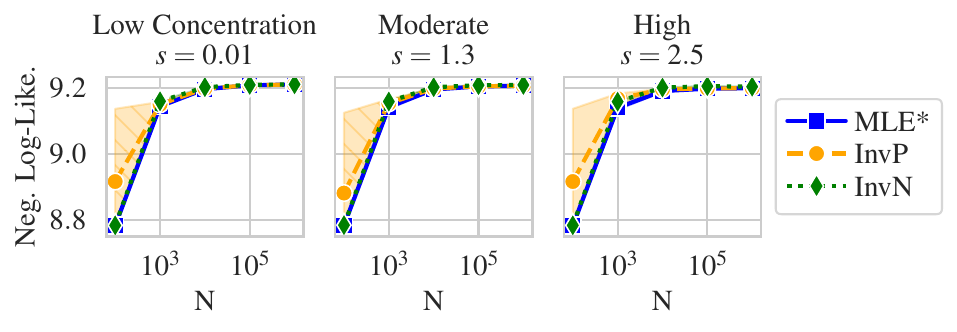}
        \caption{Negative Log-Likelihood, $K=10000, \eps=4$}
        \label{fig:overview_negloglike_fix_eps}
    \end{subfigure}

    \caption{Performance of \ourMLE, \INVP, and \INVN{} across different data distributions ($s \in \{0.01, 1.3, 5.0\}$) for fixed $\K=10{,}000$ and $\epsilon=4.0$.}
    \label{fig:overview_results_fixed_eps}
\end{figure}

\begin{figure}[!htb]
    \centering

    \begin{subfigure}[b]{1\linewidth}
        \includegraphics[width=\linewidth]{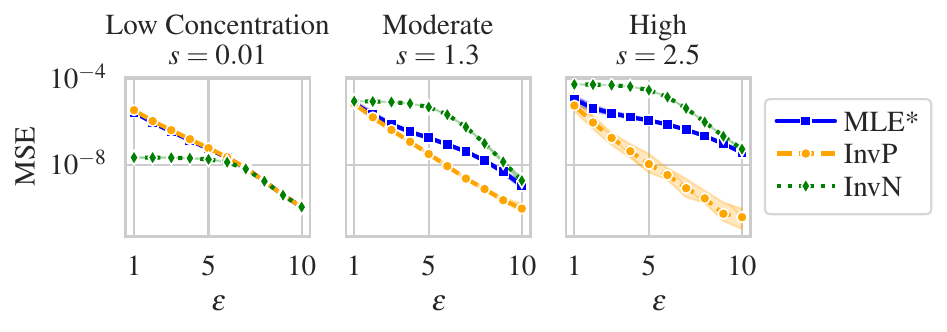}
        \caption{Mean Squared Error (MSE), $K=10000, N=10^6$}
        \label{fig:overview_mse_fix_n}
    \end{subfigure}
    \hfill
    \begin{subfigure}[b]{1\linewidth}
        \includegraphics[width=\linewidth]{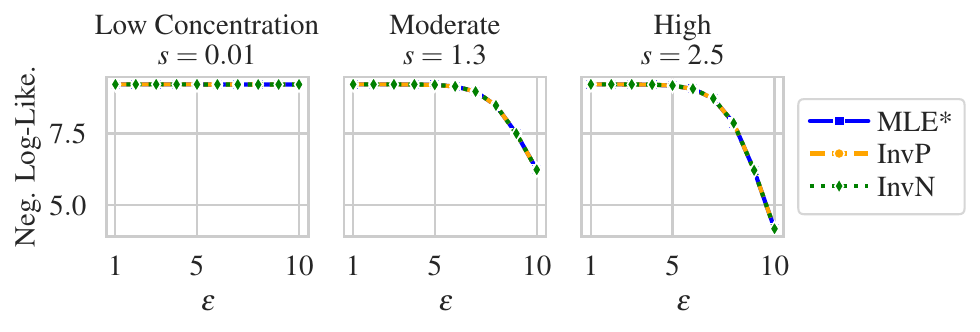}
        \caption{Negative Log-Likelihood, $K=10000, N=10^6$}
        \label{fig:overview_negloglike_fix_n}
    \end{subfigure}

    \caption{Performance of \ourMLE, \INVP, and \INVN{} across different data distributions ($s \in \{0.01, 1.3, 5.0\}$) for fixed $\K=10{,}000$ and $N=10^6$.}
    \label{fig:overview_results_fixed_n}
\end{figure}

\paragraph{Discussion.}
For the MSE metric, these results highlight the complexity of estimator behavior across different privacy and data regimes. 
As shown in Figures~\ref{fig:overview_mse_fix_eps} and~\ref{fig:overview_mse_fix_n}, the relative performance of \INVP{} and \INVN{} varies significantly with both the distributional shape (controlled by $s$) and system parameters ($N$, $\epsilon$).

We observe that in some settings, one estimator clearly dominates: for instance, \INVP{} consistently performs best under high concentration ($s=2.5$), while \INVN{} is clearly superior in the near-uniform regime ($s=0.01$). 
In other cases, the ranking between \INVP{} and \INVN{} can change depending on $\epsilon$ or $N$, with one starting off stronger but eventually being overtaken as conditions change.

Amidst this variability, \ourMLE{} demonstrates consistently strong and stable performance. 
It is consistently ``sandwiched'' between the two baselines, i.e., never the worst, often close to the best. 
This behavior underscores the \textbf{robustness} of \ourMLE{}: it adapts well across a wide range of scenarios without requiring prior knowledge of the underlying data distribution or the optimal estimator for a given configuration. 
As such, \ourMLE{} is a reliable default choice when performance must be maintained under uncertainty.

Beyond MSE, as shown in both Figures~\ref{fig:overview_negloglike_fix_eps} and~\ref{fig:overview_negloglike_fix_n}, \ourMLE{} consistently achieves the lowest Negative Log-Likelihood across all configurations, regardless of the data distribution, sample size, or privacy budget. 
This is expected, as our estimator is explicitly derived via maximum likelihood and optimized to minimize this very objective. 
In contrast, \INVP{} and \INVN{}, which are not likelihood-based, often result in poorer fit to the observed data in terms of Negative Log-Likelihood, even when they perform well under MSE. 
These results confirm that \ourMLE{} not only offers robust accuracy but is also statistically well-calibrated to the data generation process.

\begin{tcolorbox}[title=\ourMLE{} as a Robust Default Estimator]
Real-world distributions are unknown and often highly variable. Across our extensive evaluations, no single baseline estimator consistently dominates: \INVP{} and \INVN{} alternate in performance depending on the privacy budget, sample size, and distribution skew. In contrast, our likelihood-based estimator, \ourMLE{}, remains reliably close to the best in all scenarios. This consistency makes our \ourMLE{} \textbf{a safe and robust default} for practical deployment.
\end{tcolorbox}

\subsection{Experiments on Real-World Data} \label{sub:exp_real_dataset}

To assess the applicability of our estimators in practical settings, we evaluate their performance on two real-world datasets: \textsc{Kosarak}\footnote{\url{http://fimi.uantwerpen.be/data/}} and \textsc{ACSIncome}~\cite{ding2021retiring}. 
In both cases, we fix the number of users $N$ to the dataset size and vary the privacy parameter in the range $\epsilon \in \{1, \dots, 10\}$. 
We set the domain size $\K$ based on the number of unique values in the selected column, as described below.

\paragraph{Kosarak.}
This dataset consists of clickstream data from a Hungarian online news website. 
Each record represents a user and the set of URLs they clicked. 
We extract a single histogram (1st reported URL per user) over all 41{,}270 unique URLs, resulting in a domain size of $K = 41{,}270$.

\paragraph{ACSIncome.}
This dataset is derived from the US Census. 
We extract histograms based on two distinct attributes:
\begin{itemize}
    \item \textbf{PUMA:} The \emph{Public Use Microdata Area code}, a geographic identifier with codes ranging from 100 to 70{,}301. 
    We thus set the domain size to $K = 70{,}201$.
    
    \item \textbf{PINCP:} The \emph{Total Person's Income}, is a continuous variable ranging from 100 to 1{,}423{,}000. 
    We thus set the domain size to $K = 1{,}423{,}000$.
\end{itemize}

\paragraph{Evaluation.}
For each dataset and attribute, we evaluate the MSE as a function of the privacy parameter $\epsilon$. 
We also include the true underlying histogram for visualization. 
These experiments highlight the behavior of the estimators in high-dimensional, real-world scenarios, where the distributions can be highly skewed or sparse.

Figure~\ref{fig:realdata} summarizes the results. 
Each subfigure displays: (left) the true data distribution and (right) the MSE as a function of the privacy budget $\epsilon$, for the \textsc{Kosarak} dataset (a), and for the \textsc{ACSIncome} dataset using the \texttt{PUMA} (b) and \texttt{PINCP} (c) attributes, respectively.

\paragraph{Discussion.}
The findings from real-world datasets further reinforce the insights observed in our controlled synthetic experiments (Section~\ref{sub:exp_main_results}). 
As shown in Figure~\ref{fig:realdata}, the \textsc{ACSIncome} distributions (\texttt{PUMA} and \texttt{PINCP}) exhibit moderate concentration, similar to the synthetic setting with Zipf parameter $s = 1.3$. 
In contrast, the \textsc{Kosarak} dataset shows a highly peaked and sparse distribution, closely resembling the highly concentrated regime with $s \geq 2.5$.

In these respective settings, \INVP{} performs best under strong concentration (as in \textsc{Kosarak}), while \INVP{} and \INVN{} alternate in the moderately skewed regime (\textsc{ACSIncome}). 
It is important to note, however, that these observations are based on fixed domain size $\K$ and a fixed number of users $N$. 
Despite such fluctuations in relative performance, \ourMLE{} consistently remains close to the best-performing estimator across all cases. 
Its stable behavior across datasets, metrics, and privacy levels reinforces our central conclusion: \ourMLE{} is a robust and dependable choice, particularly in practical settings where the true data distribution is unknown and must be inferred from obfuscated data using the RR mechanism.

\begin{figure}[!htb]
    \centering
    \begin{subfigure}[b]{1\linewidth}
        \includegraphics[width=\linewidth]{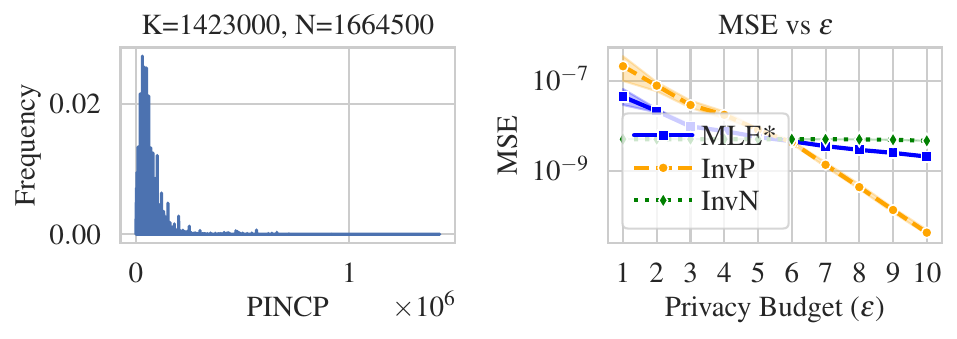}
        \caption{ACSIncome (PINCP)}
    \end{subfigure}
    \begin{subfigure}[b]{1\linewidth}
        \includegraphics[width=\linewidth]{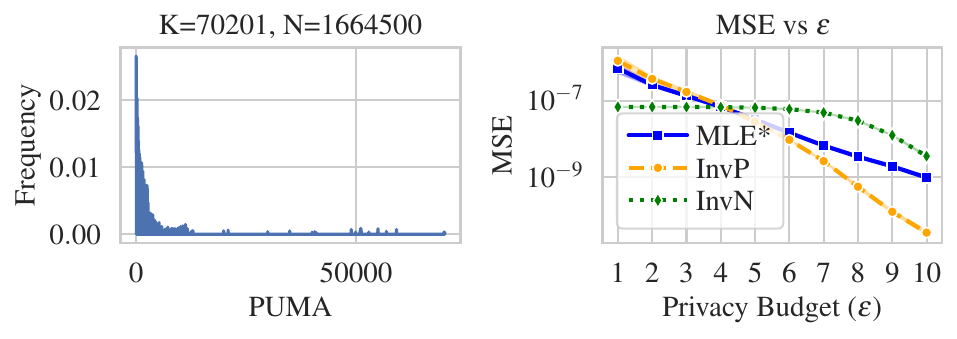}
        \caption{ACSIncome (PUMA)}
    \end{subfigure}
    \begin{subfigure}[b]{1\linewidth}
        \includegraphics[width=\linewidth]{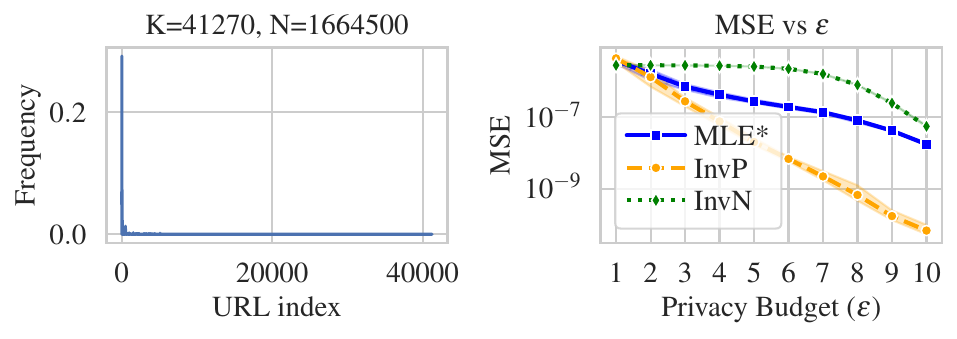}
        \caption{Kosarak}
    \end{subfigure}
    \caption{Performance of \ourMLE, \INVP, and \INVN{} across three different real-world distributions.}
    \label{fig:realdata}
\end{figure}

\subsection{The exact MLE vs. \IBU's approximation} 
\label{sub:mle_vs_ibu_convergence}
In the following, we experimentally compare our \ourMLE, which gives the 
exact MLE for the original distribution, and the existing \IBU method, which iteratively 
approximates this estimate. This comparison verifies the correctness of \ourMLE empirically, 
and also shows the computational savings that it achieves. Precisely, we measure the 
squared error between \IBU's estimate at each iteration and the exact result of  
\ourMLE. Figure \ref{fig:ibu_vs_mle}, shows the plot of this error in four 
experiments performed with two original Zipf distributions ($s = 0.01, 1.3$), two sample sizes ($n = 10^4, n= 50\times10^4$), and different privacy levels ($\epsilon = 0.5, 1.0, 2.0$). In each experiment, we sample $n$ data points from the original distribution, sanitize these samples using an RR mechanism (with $\K=500$), and finally run \IBU to approximate the MLE of the 
original distribution. 

It can be seen that the error of \IBU converges to $0$ with more iterations, hence 
confirming the correctness of \ourMLE empirically. The speed of this convergence depends 
on the setting as follows. For $n=10^4$, and a moderate level of privacy, 
e.g. $\epsilon = 1.0$, \IBU requires about $7000$ iterations to approach the exact value obtained via \ourMLE, 
while with a stronger level of privacy ($\epsilon=0.5$), 
it requires a larger number of iterations (around $25000$) to obtain the same approximation level. 
For the larger population ($n=50 \times 10^4$), many more iterations are needed for \IBU to reach a reliable approximation. In particular, the moderate privacy level $\epsilon = 1.0$ requires about $20000$ iterations, 
while for the stronger privacy ($\epsilon=0.5$), even $40000$ iterations are not enough to 
reach the same approximation precision. This indeed shows the computational superiority of \ourMLE.

%
\begin{figure}[!htb]
    \centering
    \begin{subfigure}[b]{0.49\linewidth}
        \includegraphics[width=1\linewidth]{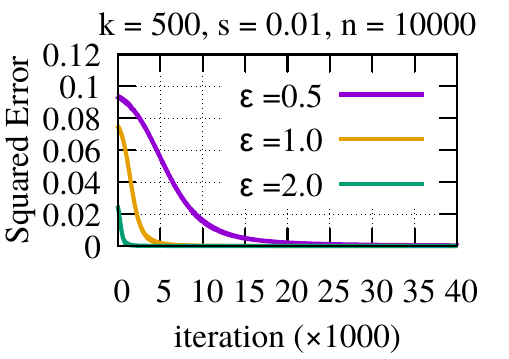}
    \end{subfigure}
    \begin{subfigure}[b]{0.49\linewidth}
        \includegraphics[width=1\linewidth]{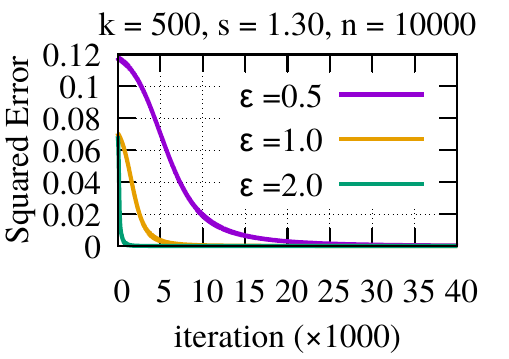}
    \end{subfigure}
    \begin{subfigure}[b]{0.49\linewidth}
        \includegraphics[width=1\linewidth]{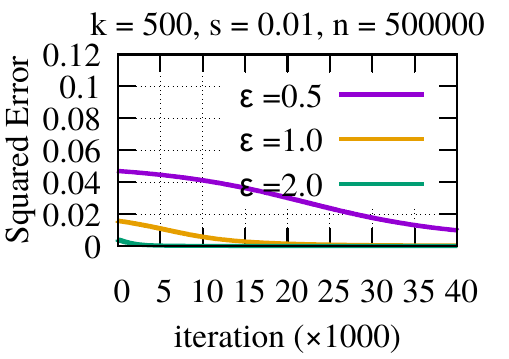}
    \end{subfigure}
    \begin{subfigure}[b]{0.49\linewidth}
        \includegraphics[width=1\linewidth]{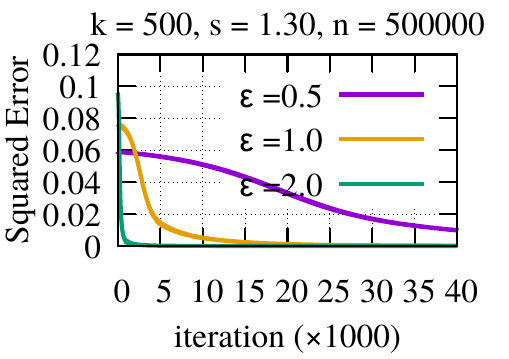}
    \end{subfigure}
    \caption{The squared error of the \IBU's estimates relative to \ourMLE, with the number of iterations.}
    \label{fig:ibu_vs_mle}
\end{figure}

\section{Conclusion} \label{sec:conclusion}

We addressed several limitations in frequency estimation under randomized response.

First, we derived a formula for the Maximum Likelihood Estimator and proved its correctness both theoretically and empirically. We also proved the uniqueness of the MLE.
Our derivation corroborates the results of \cite{kairouz2016discrete}, Supplementary Material, Section F, and, since the formula does not involve limits of any kind, it can be used to better understand the analytical properties of the MLE.
Second, we proposed an algorithm, \ourMLE, that computes the MLE in $O(\K \log \K)$ using the formula we found, vastly outperforming the iterative procedure IBU in speed while also providing an exact output as opposed to an approximation.
Third, we conducted extensive empirical comparisons of the most prominent valid estimators (\INVN, \INVP and the MLE using \ourMLE) to shed light on their behavior and trade-offs under different scenarios.
We found that in terms of mean squared error, the projection method \INVP outperforms \INVN on average when the unknown target distribution is highly concentrated and vice-versa.
Furthermore, while \INVP and \INVN outperform each other under specific conditions, \ourMLE consistently stays in between, which makes it robust in that it is never the worst of the three.
Its validity, robustness, and efficiency make \ourMLE a reliable default choice when the data distribution is unknown and performance under uncertainty matters.

Since many practical mechanisms are built upon or inspired by RR, an interesting direction for future research is to extend our analysis of the formula of \ourMLE{} to other LDP frequency estimation protocols~\cite{wang2017locally}, in particular to \emph{longitudinal LDP protocols}~\cite{erlingsson2014rappor,arcolezi2022frequency,ding2017collecting}, which protect users' data across repeated collections and are highly relevant in real-world telemetry and monitoring applications.
Another future direction is to combine different estimators to enrich the analysis. Finally, a bias-variance is included in the code repository~\cite{ourRepository}.


\bibliography{main.bib}

\newpage\input{supplementary-material} 

\end{document}

%% file: supplementary-material.tex
\begin{center}
{\fontsize{24}{48}\selectfont \bfseries Supplementary Material}
\end{center}
\vspace{2em} 

\renewcommand{\thesection}{\Alph{section}}
\setcounter{section}{0}
\setcounter{secnumdepth}{2}

\section{Code for reproducibility}\label{SM:sec:repository}

The implementation of the estimators and the code needed to reproduce all the results and plots presented in this paper can be found in the Zip file of Supplementary Material.


For quick reference, a simple implementation in Python (zero-indexed) of \ourMLE and \IBU is presented below.

{\fontsize{8}{8}\selectfont
\begin{verbatim}
import numpy as np

def MLE(phi, p: float):
    phi = np.asarray(phi)
    K = len(phi); q = (1-p) / (K-1)
    assert p >= q and q >= 0
    sig = np.argsort(phi)
    i = 0; s = 1
    while i<K and q*s > (1-i*q) * phi[sig[i]]:
        s -= phi[sig[i]]; i += 1
    lam = (p - q) * s / (1 - i*q)
    theta = np.zeros(K)
    theta[sig[i:]] = phi[sig[i:]]/lam - q/(p-q)
    return theta

def IBU(phi, p: float, n_iters=40000):
    phi = np.asarray(phi)
    K = len(phi); q = (1-p) / (K-1)
    d = p - q
    theta = np.ones(K) / K
    for i in range(n_iters):
        Z = phi / (q + d * theta)
        theta = theta * (q * np.sum(Z) + d * Z)
    return theta
\end{verbatim}
}

\section{Proof and derivation of the formula}

\def\M{{\mathcal M}}

Let $\Delta$ be the set of all probability distributions over $\K$ categories, i.e., the simplex $\Delta = \set{\th\in\bbR^\K: \sum_i \th_i = 1, \th_i\geq 0}$.

For $\th\in\Delta$, the vector $\M(\th)\in\Delta$ denotes the expected compound distribution of the output $Y\in\set{1..\K}$ of the randomized-response mechanism $\M$ (governed by some fixed $p$ and $q$) when the input $X$ is sampled from $\th$, i.e., if $X\sim\th$ and $Y\sim\M(X)$, then $\M(\th)_y = p_Y(y) = \sum_x \th_x\; p_{Y|x}(y) = q + (p-q) \th_y$.

If $\ph\in\Delta$ is the empirical histogram of $N$ observed outputs $y_i\in\{1..\K\}$ coming from $\M$, then the \emph{likelihood} of some $\th\in\Delta$ being the original input distribution given $\ph$ is $L(\th|\ph) := p(\ph | \th) := \prod_{i=1}^N p_Y(y_i)$.
If $\th$ maximizes the likelihood in $\Delta$, we say that $\th$ is a \emph{Maximum Likelihood Estimator} (MLE) for $\ph$.

For $\ph,\th\in\Delta$ the \emph{Kullback-Leibler divergence} is given by $D_{KL}(\ph, \th) := \sum_{i=1}^\K \ph_i \log\left(\frac{\ph_i}{\th_i}\right)$, with convention $0\log\frac{0}{0}=0$.

To begin, we prove that Maximizing likelihood is equivalent to minimizing KL divergence.

\begin{theorem}\label{SM:thm:DKL}
    Maximizing likelihood is equivalent to minimizing KL divergence in the following sense:
    \[
    \setargmax_{\th\in\Delta} L(\th|\ph) = \setargmin_{\th\in\Delta} D_{KL}(\ph, \M(\th))
    \]
\end{theorem}

\begin{proof}
    For fixed $\ph$, let $f(\th) := \sum_y \ph_y \log \M(\th)_y$.
    Notice that $D_{KL}(\ph, \M(\th)) = \sum_y \ph_y \log \ph_y - f(\th)$ and
    \[
    \begin{aligned}
        &L(\th | \ph) = \prod_{i=1}^N p_Y(y_i) = \prod_{y=1}^{\K} p_Y(y)^{|\set{i:y_i{=}y}|}\\
        &= (\prod_y p_Y(y)^{\ph_y})^N = \exp\left( N \sum_y \ph_y \log p_Y(y)\right)\\
        &= \exp\left( N\, f(\th) \right)\\
    \end{aligned}
    \]
    So, $D_{KL}(\ph, \M(\th))$ is minimized if and only if $f(\th)$ is maximized, which is equivalent to maximizing $L(\th|\ph)$.
\end{proof}

We now use the Lagrange Multipliers method to derive a formula that characterizes any given MLE.

\begin{theorem}\label{SM:thm:lagrange}
    Let $\ph,\th\in\Delta$ and let $\lambda := \frac{(p-q) \sum_{\th_i>0} \ph_i}{1- |\set{i:\th_i{=}0}| q}$.
    If $\th$ minimizes $D_{KL}(\ph, \M(\th))$ in $\Delta$, then, for every $i$, either $\th_i=0$ or $\th_i = \frac{\ph_i}{\lambda} - \frac{q}{p-q}$.
\end{theorem}

\begin{proof}
To apply the theorem of the Lagrange multipliers method, we will split the constraint $\th_i\geq 0$ into $\th_i=0$ for some $i$ and $\th_i=e^{u_i}$ for some others in the following sense.
List the indices $\alpha_1,...,\alpha_m$ (here $m<\K$) and $\beta_1,...,\beta_{\K-m}$ with $\th_{\alpha_i}>0$ and $\th_{\beta_i}>0$.
For every vector $u\in \bbR^m$, let $\th = f(u)\in\bbR^\K$ be given by $f(u)_{\alpha_i} = e^{u_i}$ and $f(u)_{\beta_i} = 0$.
\def\thai{{\th_{\alpha_i}}}

The minimizer $u^*$ of $\ell(u):=\DKL(\ph, \th)$ constrained to $\sum_i \th_i = 1$ can be found with the Lagrange multipliers method: $(u^*, \lambda^*)$ must be a critical point of the Lagrange function $\mathcal L (u, \lambda) = \ell(u) - \lambda (1 - \sum_i \th_i)$.
Notice that
\[
 \frac{\partial \mathcal L}{\partial u_i} = \frac{\partial \mathcal L}{\partial\thai} \frac{\partial \thai}{\partial u_i}
 = \left(-\frac{\ph_{\alpha_i} (p-q)}{q+(p-q) \thai} + \lambda \right)\thai,
\]
and $\frac{\partial \mathcal L}{\partial \lambda} = 1 - \sum_i \th_i$.
Therefore, from $\nabla_{u,\lambda}\mathcal L(u, \lambda)|_{u^*,\lambda^*} = 0$ and the fact that $\thai>0$ by definition, it follows that the minimizer $\th:=f(u^*)$ is given by $\th_{\alpha_i} = \frac{\ph_{\alpha_i}}{\lambda^*} - \frac{q}{p-q}$ and $\th_{\beta_i} = 0$, and $\lambda^* = \frac{(p-q) \sum_{\th_i>0} \ph_i}{(p-q)+|\set{i:\th_i{>}0}|q}$.
Finally, from $p+(\K-1)q=1$, the denominator becomes simply $1- |\set{i:\th_i{=}0}| q$.
\end{proof}

With Theorem~\ref{SM:thm:lagrange}, the problem of computing all the components of $\hat\th$ is reduced into identifying the components $i$ in which $\hat\th_i=0$ because $\lambda$ and all the other components can be computed with the given formulas.
The rest of the proof is devoted to identifying the zero-valued components.

The problem might be reduced even further because, as shown in Theorem~\ref{SM:thm:permutations}, permuting $\ph$ simply permutates the MLE.
Hence, if we knew how to find an MLE for every sorted $\ph$, we could derive a method to find an MLE for every $\ph$ simply by permutating the input and then the output in the same way.

\begin{theorem}\label{SM:thm:permutations}
    (Permutations can be abstracted)
    Fix some $\ph\in\Delta$ satisfying $\ph_1\leq\cdots\leq\ph_\K$, some permutation $\sigma:\set{1..\K}\to\set{1..\K}$, and let $f_\sigma:\Delta\to\Delta$ be the permutation function $f_\sigma(\th)_i = \th_{\sigma(i)}$.
    If $\th$ is an MLE for $\ph$, then $f_\sigma(\th)$ is an MLE for $f_\sigma(\ph)$ (and vice-versa via $f_{\sigma^{-1}}$).
\end{theorem}
\begin{proof}
    Notice that $\M(f_\sigma(\th))_i = q + (p-q) f_\sigma(\th)_i = q + (p-q) \th_{\sigma(i)} = \M(\th)_{\sigma_i}$.
    Therefore,
    \[
    \begin{aligned}
        D_{KL}(f_\sigma(\ph), \M(f_\sigma(\th))) &= \sum_i \ph_{\sigma(i)} \log\frac{\ph_{\sigma(i)}}{\M(\th)_{\sigma(i)}}\\
        &= \sum_j \ph_j \log\frac{\ph_j}{\M(\th)_j}\\
        &= D_{KL}(\ph, \M(\th)).
    \end{aligned}
    \]
\end{proof}

Furthermore, as shown in Lemma~\ref{SM:lemma:monotonicity} and Theorem~\ref{SM:thm:monotonicity}, the entries in $\hat\th$ are monotonic with respect to the entries in $\ph$.

\begin{lemma}\label{SM:lemma:monotonicity}
    (Monotonicity)
    If two distributions $\th, \th'\in\Delta$ differ only at two swapped positions $j$ and $k$ with $a:=\th_j=\th'_k < b:=\th_k=\th'_j$, and $\ph_j<\ph_k$ then $D_{KL}(\ph, \M(\th)) < D_{KL}(\ph, \M(\th'))$.
\end{lemma}

\begin{proof}
Let $\delta := D_{KL}(\ph, \M(\th)) - D_{KL}(\ph, \M(\th'))$.
\[
\begin{aligned}
    \delta &= \sum_i \ph_i \log\frac{q + (p-q)\th'_i}{q + (p-q) \th_i}\\
    &= \ph_j \log\frac{q + (p-q)b}{q + (p-q)a} + \ph_k \log\frac{q + (p-q)a}{q + (p-q)b}\\
    &= (\ph_j - \ph_k) \log\frac{q + (p-q)b}{q + (p-q)a}.
\end{aligned}
\]
Since $\ph_j<\ph_k$ and $b>a$, we have $\delta< 0$.
\end{proof}

\begin{theorem}\label{SM:thm:monotonicity}
    (Monotonicity)
    If $\th$ is an MLE for $\ph$, then for any $i,j\in\set{1..\K}$, if $\ph_i\leq\ph_j$ then $\th_i\leq\th_j$.
\end{theorem}
\begin{proof}
    This is an immediate consequence of Lemma~\ref{SM:lemma:monotonicity}.
\end{proof}

Due to Theorem~\ref{SM:thm:monotonicity}, if $\hat\th$ is an MLE for $\ph$, then the components of $\hat\th$ preserve the ordering (by value) of the corresponding components in $\ph$.
Therefore, the zero components of $\hat\th$ must align with the positions of the $n$ smallest components of $\ph$, for some unknown $n\in\set{0,...,\K-1}$.
In particular, when $\ph$ is sorted, the zeros of $\hat\th$ must occur in the first $n$ coordinates.
This largely reduces the search space for $\hat\th$, because if we let $\th\nth$ be the result of applying the formula of Theorem~\ref{SM:thm:lagrange} (Definition~\ref{SM:def:nth}), then it must hold that $\hat\th$ equals $\th\nth$ for some $n$ (Lemma~\ref{SM:lemma:optimal-is-nth}).

\begin{definition}\label{SM:def:nth}
    Fix some $\ph\in\Delta$ satisfying $\ph_1\leq\cdots\leq\ph_\K$, and some $n\in\set{0..\K-1}$.
    Define $\lambda\nth := \frac{(p-q) \sum_{i>n} \ph_i}{1-n q}$, and let $\th\nth$ be the $\K$-dimensional vector given by\[
        \th\nth_i := \begin{cases}
        0 &\text{ if }i\leq n\\
        \frac{\ph_i}{\lambda\nth} - \frac{q}{p-q} &\text{ otherwise.}
        \end{cases}
    \]
    Notice that $\sum_i\th\nth_i = 1$, but it is not necessarily the case that $\th\nth\in\Delta$ because $\th\nth_i < 0$ may occur.
    Notice also that \[
    \M(\th\nth)=
    \begin{cases}
        q &\text{ if }i\leq n\\
        \frac{(p-q)\ph_i}{\lambda\nth}&\text{ otherwise.}        
    \end{cases}
    \]
\end{definition}

\begin{lemma}\label{SM:lemma:optimal-is-nth}
    Fix some $\ph\in\Delta$ satisfying $\ph_1\leq\cdots\leq\ph_\K$.
    If $\th$ is an MLE for $\ph$, then there exists some $n\in\set{0,...,\K-1}$ such that $\th = \th\nth$.
\end{lemma}

\begin{proof}
    By Theorem~\ref{SM:thm:monotonicity}, we have $\th_1\leq\cdots\leq\th_\K$.
    Let $n = |\set{i:\th_i=0}|$.
    Notice that $\th_i=0$ iff $\th\nth_i = 0$, because ($\rightarrow$) if $\th_i=0$, then $i\leq n$ and $\th\nth_i = 0$, and ($\leftarrow$) if $\th_i>0$, then $i>n$ and $\th\nth_i>0$.
    From Theorem~\ref{SM:thm:lagrange}, it follows that $\th_i = \frac{\ph_i}{\lambda}-\frac{q}{p-q} =\th\nth_i$ for all $i>n$.
\end{proof}

In order to determine the value of $n$ for which the MLE is $\th\nth$ (still assuming sorted $\ph$), we first prove some properties of $\th\nth$ for arbitrary $n$ in Lemma~\ref{SM:lemma:props-nth}.

\begin{lemma}\label{SM:lemma:props-nth}
    Fix some $\ph\in\Delta$ satisfying $\ph_1\leq\cdots\leq\ph_\K$, and some $n\in\set{0..\K-1}$.
    Let $g(n):= (1-nq)\ph_{n+1} - q (\sum_{i>n} \ph_i)$.
    The following assertions hold:
    \begin{enumerate}
    \setlength{\topsep}{0pt}
    \setlength{\itemsep}{0pt}
    \setlength{\parskip}{0pt}
    \setlength{\itemindent}{-0.5em} 
    \setlength{\labelsep}{0.2em} 
        \item[(a)] $\th\nth \in \Delta$ if and only if $g(n)\geq 0$.
        \item[(b)] For $n>0$, if $\th\nthprev \in \Delta$, then $\th\nth \in \Delta$.
        \item[(c)] For $n>0$, if $\th\nthprev \in \Delta$, then $\lambda\nthprev \geq \lambda\nth$.
        \item[(d)] For $n>0$, if $\ph_{n-1}=\ph_n$ then $g(n)=g(n-1)$.
        \item[(e)] For $n>0$, if $g(n-1)=0$ then $\theta\nthprev=\theta\nth$.
    \end{enumerate}
\end{lemma}
\begin{proof}
(a)  Let $\th=\th\nth$ and $\lambda=\lambda\nth$.
We have $\th\in\Delta$ iff $\frac{\ph_i}{\lambda} \geq \frac{q}{p-q}$ for all $i>n$.
This holds for all $i>n$ iff it holds for $i=n+1$ since $\ph_{n+1}\leq\ph_{n+2}\leq\cdots\leq\ph_\K$.
So, \[
\begin{aligned}
    \th\in\Delta &\iff \frac{\ph_{n+1}}{\lambda} \geq \frac{q}{p-q}\\
    &\iff \ph_{n+1} (1-nq) \geq q \Sigma_{i>n} \ph_i\\
    &\iff g(n) \geq 0.
\end{aligned}
\]
(b) Since $\ph_{n+1}\geq\ph_n$,
\[
\begin{aligned}
    g(n)-g(n-1) &= (1-nq)\ph_{n+1} - (1-(n-1)q)\ph_n\\
    &\phantom{=} - q(\Sigma_{i>n} \ph_i - \Sigma_{i>n-1} \ph_i)\\
    &= (1-nq) (\ph_{n+1} - \ph_n) \geq 0\\
\end{aligned}
\]
(c) We have $g(n-1)\geq 0$ from $\th\nthprev\in\Delta$, and
\[
\begin{aligned}
    &\sign(\lambda\nthprev - \lambda\nth)\\
    &= \sign(\frac{\Sigma_{i>n-1} \ph_i}{1-nq+q} - \frac{\Sigma_{i>n} \ph_i}{1-nq})\\
    &= \sign(\underbrace{(1-nq)(\Sigma_{i>n-1} \ph_i) - (1-nq+q)(\Sigma_{i>n} \ph_i)}_\delta),
\end{aligned}
\]
and
\[
\begin{aligned}
    \delta &= (1-nq) \ph_n - q\sum_{i>n} \ph_i\\
    &= (1-(n-1)q) \ph_n - q\sum_{i>n-1} \ph_i = g(n-1)\geq 0,
\end{aligned}
\]
so $\lambda\nthprev \geq \lambda\nth$.
\\(d) If $\ph_{n+1} = \ph_n$, from (b) one obtains $g(n)-g(n-1) = (1-nq) (\ph_{n+1} - \ph_n) = 0$.
\\(e) If $g(n-1) = 0$ then $\delta=0$ and thus $\lambda\nthprev=\lambda\nth$ which means that $\th\nthprev=\th\nth$.

\end{proof}

We now characterize the value of $n$ using the aforementioned properties.

\begin{lemma}\label{SM:lemma:n-smallest}
    ($n$ is as small as possible)
    For $n>0$, if $\th\nthprev\in\Delta$ then $\th\nth\in\Delta$ and $D_{KL}(\ph, \M(\th\nthprev)) \leq D_{KL}(\ph, \M(\th\nth))$.
\end{lemma}
\begin{proof}
    Let $\delta_i:=\ph_i \log\frac{\M(\th\nthprev)_i}{\M(\th\nth)_i}$.
    Notice that $\delta_i\geq 0$ for all $i>n$, because
    \begin{itemize}
        \item For $i < n$, $\M(\th\nthprev)_i =q = \M(\th\nth)_i$, so $\delta_i=0$.
        \item For $i=n$, since $\th\nthprev\in\Delta$, when $\th\nthprev_i\geq 0$, hence $\M(\th\nthprev)_i \geq q = \M(\th\nth)_i$.
        \item For $i>n$, $\delta_i = \ph_i \log \frac{\lambda\nthprev \ph_i}{\lambda\nth \ph_i}$, in which case $\delta_i\geq 0$ because $\lambda\nthprev\geq\lambda\nth$ (Lemma~\ref{SM:lemma:props-nth}).
    \end{itemize}
    Therefore, for $\delta:=D_{KL}(\ph, \M(\th\nth)) - D_{KL}(\ph, \M(\th\nthprev))$, we have
    $\delta = \sum_i \delta_i \geq 0$.
\end{proof}

Finally, everything is put together in Theorem~\ref{SM:thm:new-main}.

\begin{theorem}\label{SM:thm:new-main}
    (Wrap up)
    Fix any given distribution $\ph'\in\Delta$ and consider the following procedure.
    \begin{enumerate}
    \setlength{\itemsep}{0pt}
    \setlength{\parskip}{0pt}
    \setlength{\itemindent}{-0.5em} 
    \setlength{\labelsep}{0.2em} 
    \item[1.] Let $\sigma$ be an arg-sort permutation for $\ph'$, so that $\ph'_{\sigma(1)}\leq\cdots\leq\ph'_{\sigma(\K)}$.
    \item[2.] Let $\ph\in\Delta$ be given by $\ph_i:=f_\sigma(\ph)_i=\ph'_{\sigma(i)}$.
    \item[3.] Let $n$ be the smallest integer such that $g(n)\geq 0$, where $g(n):= (1-nq)\ph_{n+1} - q (\sum_{i>n} \ph_i)$.
    \item[4.] Let $\th\nth$ be the vector defined in Definition~\ref{SM:def:nth} for $\ph$ and $n$.
    \item[5.] Let $\th'\in\Delta$ be given by $\th'_{\sigma(i)} = \th\nth_i$.
    \end{enumerate}
    Then $\th'$ is the unique MLE for $\ph'$.
    Alternatively, for pedagogic reasons, $\th\nth$ can also be presented as \[
    \th\nth_i := \begin{cases}
        0 &\text{ if } \ph_i\leq\tau := \ph_n
        \\ \frac{\ph_i}{\lambda} - \frac{q}{p-q}  &\text{ otherwise}.
    \end{cases}
    \]
\end{theorem}

\begin{proof}
    By Theorem~\ref{SM:thm:permutations}, the search for an MLE $\theta'$ for $\ph'$ can be reduced to the search for an MLE $\th$ for $\ph$, via $\th_i=\th_{\sigma(i)}$.
    Since $\ph$ is sorted, by Lemma~\ref{SM:lemma:optimal-is-nth}, there exists some $n\in\set{0..\K}$ such that $\th = \th\nth$.
    By Lemma~\ref{SM:lemma:n-smallest}, the set $\set{n:\th\nth\in\Delta}$ is a non-empty suffix of $\set{0..\K-1}$, and the optimal $n$ is the smallest in this set.
    Regarding the alternative presentation, Lemmas~\ref{SM:lemma:n-smallest} and~\ref{SM:lemma:props-nth} (Part (d)) imply that the optimal $n$ satisfies $n=0$ or $\ph_{n-1}\ne\ph_n$.
    This means that $i\leq n$ iff $\ph_i\leq \tau := \ph_n$, so the two definitions are equivalent.
\end{proof}

\section{Theoretical use of the formula}\label{SM:app:collinearity}

\begin{theorem}\label{SM:thm:collinearity}
    If the smallest entry in $\phi$ takes value $v<q$, and the second smallest is at least $(\K q - v) / (\K - 1)$, then $\hat\th^\INVP$, $\hat\th^\ourMLE$ and $\hat\th^\INVN$ are collinear, with $\hat\th^\ourMLE$ in the middle.
\end{theorem}
\begin{proof}
Let $n$ be the coordinate of the smallest component $\ph_n<q$.
We have $\hat\th^\INVP_n = \hat\th^\INVN_n = \hat\th^\ourMLE_n = 0$, and the remaining components are given by the following formulas.
\[
\begin{aligned}
\hat\th^\INVP_i &= \ph_i \frac{1}{p-q} - \frac{(\K q - \ph_n)}{(\K - 1) \pq} \\
\hat\th^\INVN_n &= \ph_i \frac{1}{p-\ph_n} - \frac{q}{p-\ph_n} \\
\hat\th^\ourMLE_n &= \ph_i \frac{1-q}{\pq (1-\ph_n)} - \frac{q}{\pq}
\end{aligned}
\]

Letting $t:=\frac{1-\ph_n}{p-\ph_n}$, it can be shown that $\hat\th^\INVN-\hat\th^\INVP = t (\hat\th^\ourMLE-\hat\th^\INVP)$.
Therefore, the three estimators are collinear.
Furthermore, since $1-\ph_n > p-\ph_n$, then $t>1$, which implies that $\hat\th^\ourMLE$ is between $\hat\th^\INVN$ and $\hat\th^\INVP$.
\end{proof}

\begin{lemma}\label{SM:lemma:continuity}
    The four estimators (\INV, \INVP, \INVN and \ourMLE) are continuous, i.e. the function $\phi \mapsto \hat\theta$ is continuous where $\hat\theta$ is the result under each respective estimator. 
\end{lemma}

\begin{proof}
    
    The estimator $\INV$ is a linear transformation of the distribution $\phi$, 
    and is therefore continuous. By definition, $\INVP$ can be written as 
    $\INVP(\phi) = f(\INV(\phi))$, where $f(v)$ is the projection of $v\in \mathbb{R}^K$
    onto the set of distributions over $\{1,2, \dots, K\}$. Note that $f$ is 
    a continuous mapping since that set is closed and convex. 
    Therefore $\INVP$ is continuous by the continuity of both $\INV$ and $f$. 

    $\INVN$ can be written as $\INVN(\phi) = f(\INV(\phi))$ for the
    postprocessing function $f:\mathbb{R}^K \to \mathbb{R}^K$ (that transforms a vector 
    into a valid distribution) via
    $f(v)_i = \frac{\max(0, v_i)}{\sum_{i=1}^K \max(0,v_i)}$.
    Note that the numerator and the denominator are continuous functions of the vector $v$ and the dominator is always positive, so $f$ is continuous. 
    Hence, by the continuity of $f$ and \INV, \INVN is continuous.
    
    Lastly we want to prove that $h(\phi) :=  \hat\th^\ourMLE$ is continuous. Using Theorem~\ref{SM:thm:permutations} we can suppose that the entries of $\phi$ are increasing. If we recall the notation of $g$ in Lemma~\ref{SM:lemma:props-nth} and partition the domain of $\ph$ as $\Delta = \cup_{n=0}^{\K-1} A_n$ where $A_n = \{\phi : g(n)\geq 0 \ \land g(n-1)<0\}$, then $h$ can be expressed as a piecewise continuous function taking value $\th\nth$ in each set $A_n$.
    For $n$ fixed, $h$ is continuous in $A_n$, therefore, the continuity of the MLE needs only to be tested at the borderline points $\phi$ where $g(n)=0$.
    It was shown in Lemma~\ref{SM:lemma:props-nth}~(e) that ${\th\nth}=\th\nthnext$, so $h$ is continuous in $\Delta$.
\end{proof}

\def\pto{\stackrel{\text{p}}{\to}}
\begin{lemma}\label{SM:lemma:continuity-implies-consistency}
    If $\hat\th:=h(\ph)$ is an estimator for $\th\in\Delta$ and $h$ is continuous and satisfies $h(\ph^*)=\th$, where $\ph^*_i:=\pq \th_i + q$, then $\hat\th$ is consistent.
\end{lemma}
\begin{proof}
    Suppose $h:\Delta\to\Delta$ is continuous and fix $\th\in\Delta$.
    For every $n$, let $X^{(n)}$ be a random multiset of $n$ values sampled from $\set{1..\K}$ with distribution $\th$, let $Y^{(n)}$ be the multiset of randomized values after applying RR to each value in $X^{(n)}$, and let $\ph_n\in\Delta$ be the histogram of $Y^{(n)}$.
    This forms a sequence of random variables $\ph^{(1)}, \ph^{(2)}, ...\in\Delta$.

    $Y^{(n)}$ follows a categorical distribution governed by $\Pr(Y^{(n)}=y\,|\th) = p\th_y + q (1-\th_y) = \ph^*_y$, which is independent of $n$.
    Therefore $\ph^{(n)}$ follows a a Multinomial distribution, so the law of large numbers for Multinomial distributions guarantees that $\ph^{(n)} \pto \ph^*$ (componentwise convergence in probability).

    Using the Continuous Mapping Theorem, 
    since $h$ is continuous, then $h(\ph^{(n)}) \pto h(\ph^*) = \theta$.
    In addition, $\norm{\ph^{(n)}}$ is bounded because $\ph^{(n)}\in\Delta$, so convergence in probability or in distribution to a constant is equivalent to convergence in L2.
    Therefore, $\norm{h(\ph) - \theta} \pto 0$.
\end{proof}

\begin{theorem}\label{SM:thm:consistency}
    The four estimators (\INV, \INVP, \INVN and \ourMLE) are consistent.
    Moreover, for \INV and \INVP, the rate of L2-convergence is $O(\K/N)$.
\end{theorem}

\begin{proof}
The consistency of the four estimators is an immediate consequence of Lemmas~\ref{SM:lemma:continuity} and~\ref{SM:lemma:continuity-implies-consistency}.

For $\hat \th = \hat\th^\INV$, from Equation~\eqref{eq:var-inv} and its unbiasedness, we have $\MSE_\th(\hat \th) = \sum_{i=1}^\K E((\hat \th_i-\th_i)^2) = \sum_{i=1}^\K \Var(\hat \th_i)$, which yields $\MSE_\th(\hat \th) = \frac{\K q(1{-}q) + \pq(1{-}2q{-}\pq)}{N\pq}$.

For $\hat \th=\hat\th^\INVP$, let $\hat z=\hat\th^\INV$.
By definition, we have $\norm{\hat z - \hat \th} \leq \norm{\hat z - \th}$.
So $\norm{\hat \th - \th}\leq \norm{\hat \th - \hat z} + \norm{\hat z - \th} \leq 2 \norm{\hat z - \th} \in O(\K/N)$.
\end{proof}

\section{Complete Results} \label{SM:app:results}

\paragraph{Evaluation.} 
This section reports the full set of experimental results across all configurations of domain size $\K \in \{50, 100, 1000, 5000, 10000\}$, privacy budgets $\epsilon \in \{1, 2, 4, 10\}$, and sample sizes $N \in \{10^2, \dots, 10^6\}$.
We evaluate three estimators (our \ourMLE{}, \INVP{}, and \INVN{}) under three Zipf concentration levels: $s = 0.01$ (near-uniform), $s = 1.3$ (moderate skew), and $s = 2.5$ (high concentration).

Figures~\ref{SM:fig:zipf_mse_s_0.01} to~\ref{SM:fig:zipf_mse_s_2.5} show Mean Squared Error (MSE) results, and Figures~\ref{SM:fig:zipf_nll_s_0.01} to~\ref{SM:fig:zipf_nll_s_2.5} present corresponding negative log-likelihood values.
Each figure is structured as a grid of subplots: rows vary $K$, columns vary $\epsilon$, and the leftmost column in each row displays the true underlying histogram.
This layout provides a comprehensive view of how estimator performance evolves with domain size, privacy level, and distributional skew.

\paragraph{Discussion.}
These results confirm and extend the observations made in the main paper. 
The performance of \INVP{} and \INVN{} varies significantly depending on the privacy level, sample size, and data skewness. 
In some settings, \INVP{} consistently outperforms the others (especially under high concentration), while in others, \INVN{} dominates (e.g., under uniform distributions). 
In yet other configurations, the relative ordering between these baselines changes with $N$ or $\epsilon$. 
Throughout all conditions, \ourMLE{} remains reliably close to the best-performing estimator, regardless of metric, domain size, or distribution. 
This stable and predictable behavior reinforces our conclusion that \ourMLE{} is a robust and principled default, especially in practical scenarios where the true data distribution is unknown and the privacy regime can vary.

\begin{figure*}[!htb]
    \centering
    \includegraphics[width=1\linewidth]{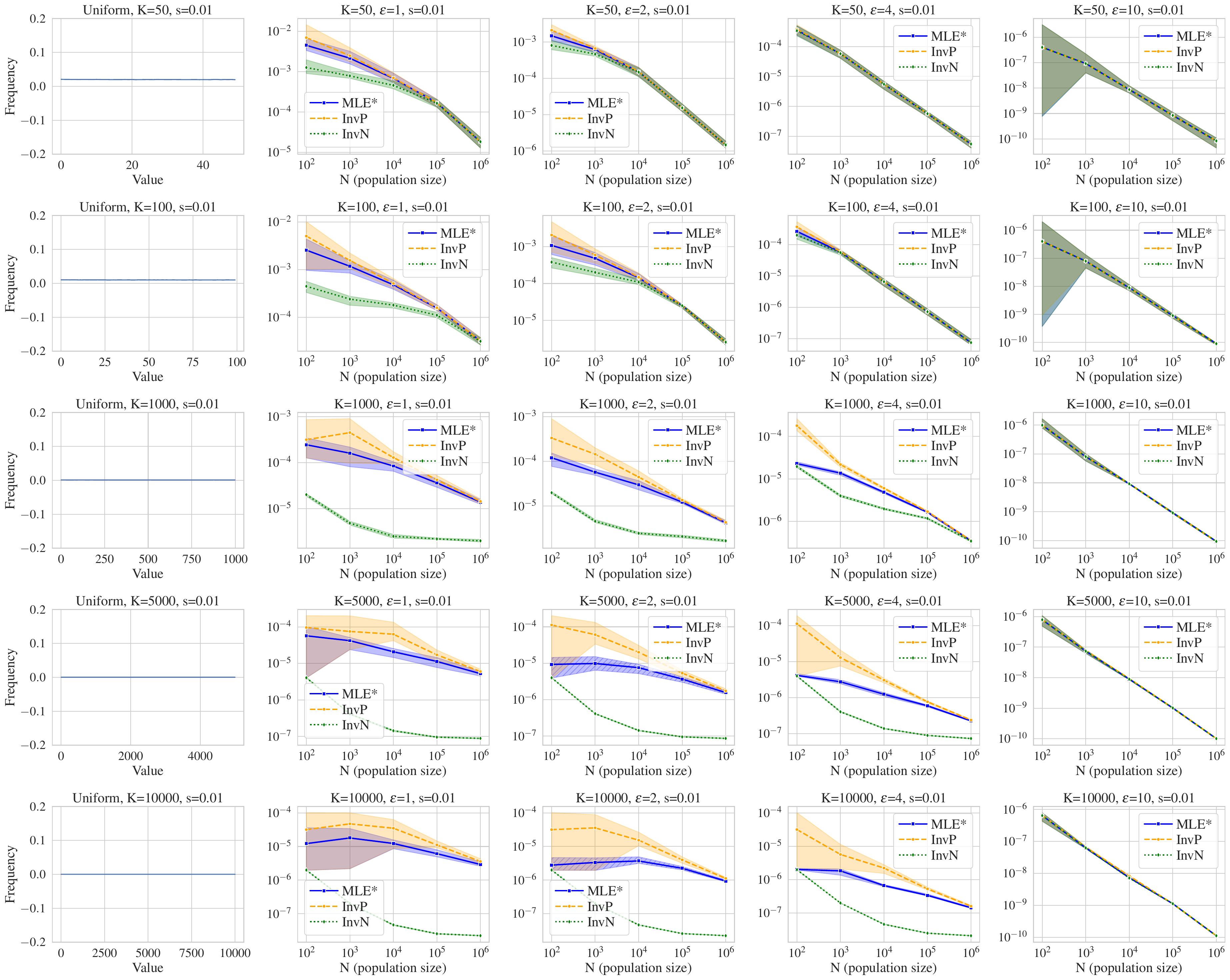}
    \caption{MSE results for near-uniform distribution ($s = 0.01$). 
    Rows vary domain size $K$, and columns vary privacy level $\epsilon$. 
    The far-left subplot of each row shows the true histogram.}
    \label{SM:fig:zipf_mse_s_0.01}
\end{figure*}

\begin{figure*}[!htb]
    \centering
    \includegraphics[width=1\linewidth]{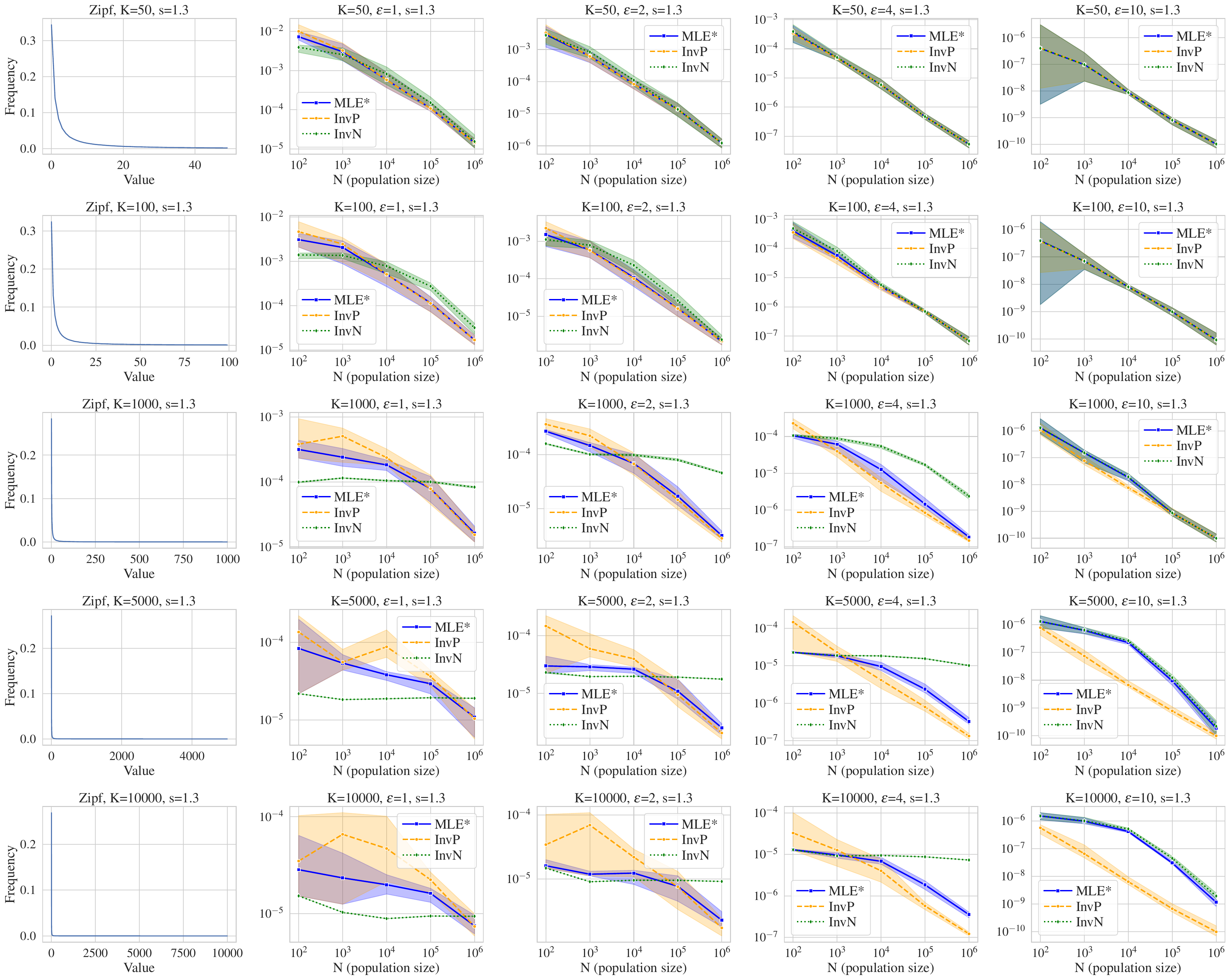}
    \caption{MSE results for moderately skewed distribution ($s = 1.3$). 
    Rows vary domain size $K$, and columns vary privacy level $\epsilon$. 
    The far-left subplot of each row shows the true histogram.}
    \label{SM:fig:zipf_mse_s_1.3}
\end{figure*}

\begin{figure*}[!htb]
    \centering
    \includegraphics[width=1\linewidth]{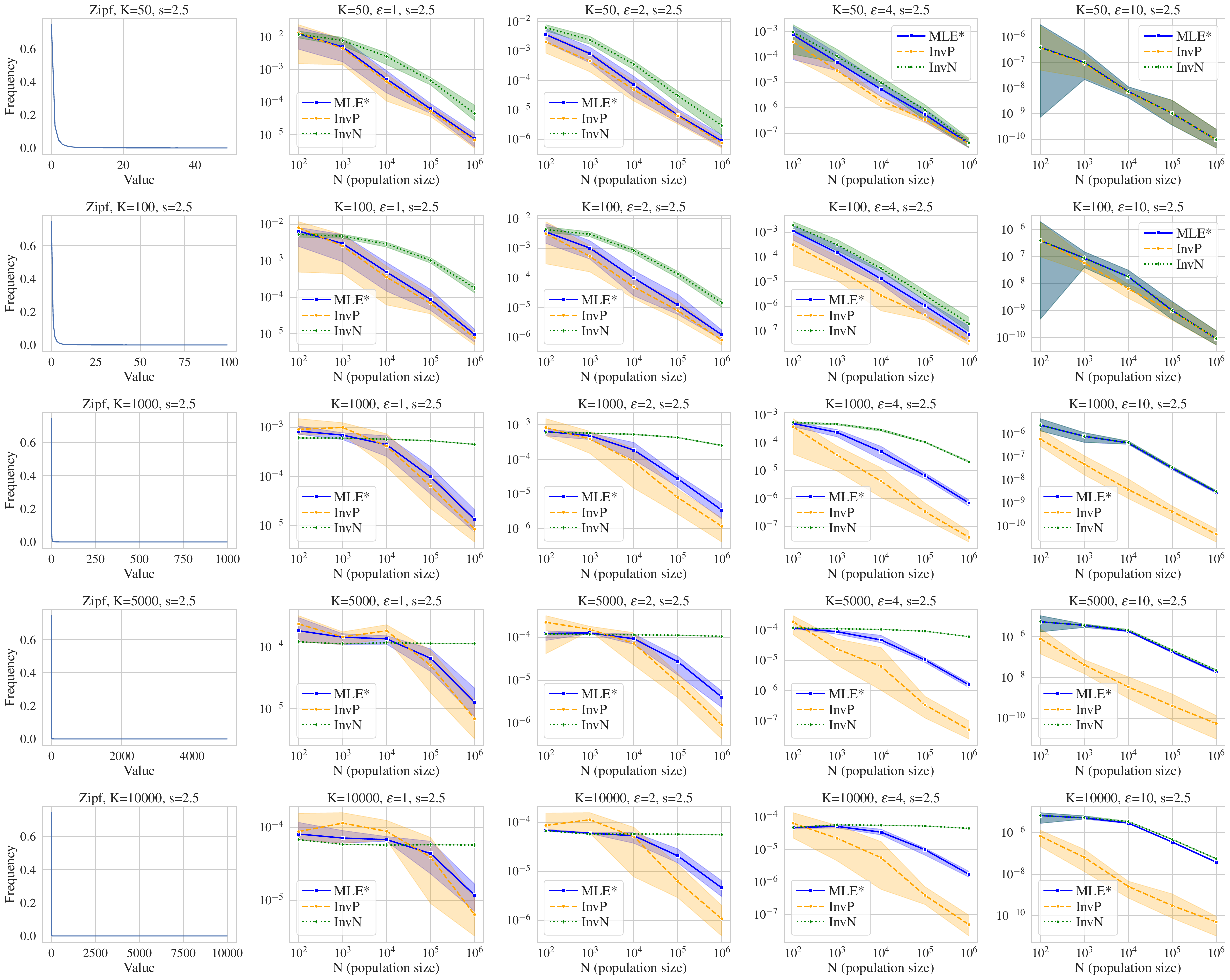}
    \caption{MSE results for highly concentrated distribution ($s = 2.5$). 
    Rows vary domain size $K$, and columns vary privacy level $\epsilon$. 
    The far-left subplot of each row shows the true histogram.}
    \label{SM:fig:zipf_mse_s_2.5}
\end{figure*}

\begin{figure*}[!htb]
    \centering
    \includegraphics[width=1\linewidth]{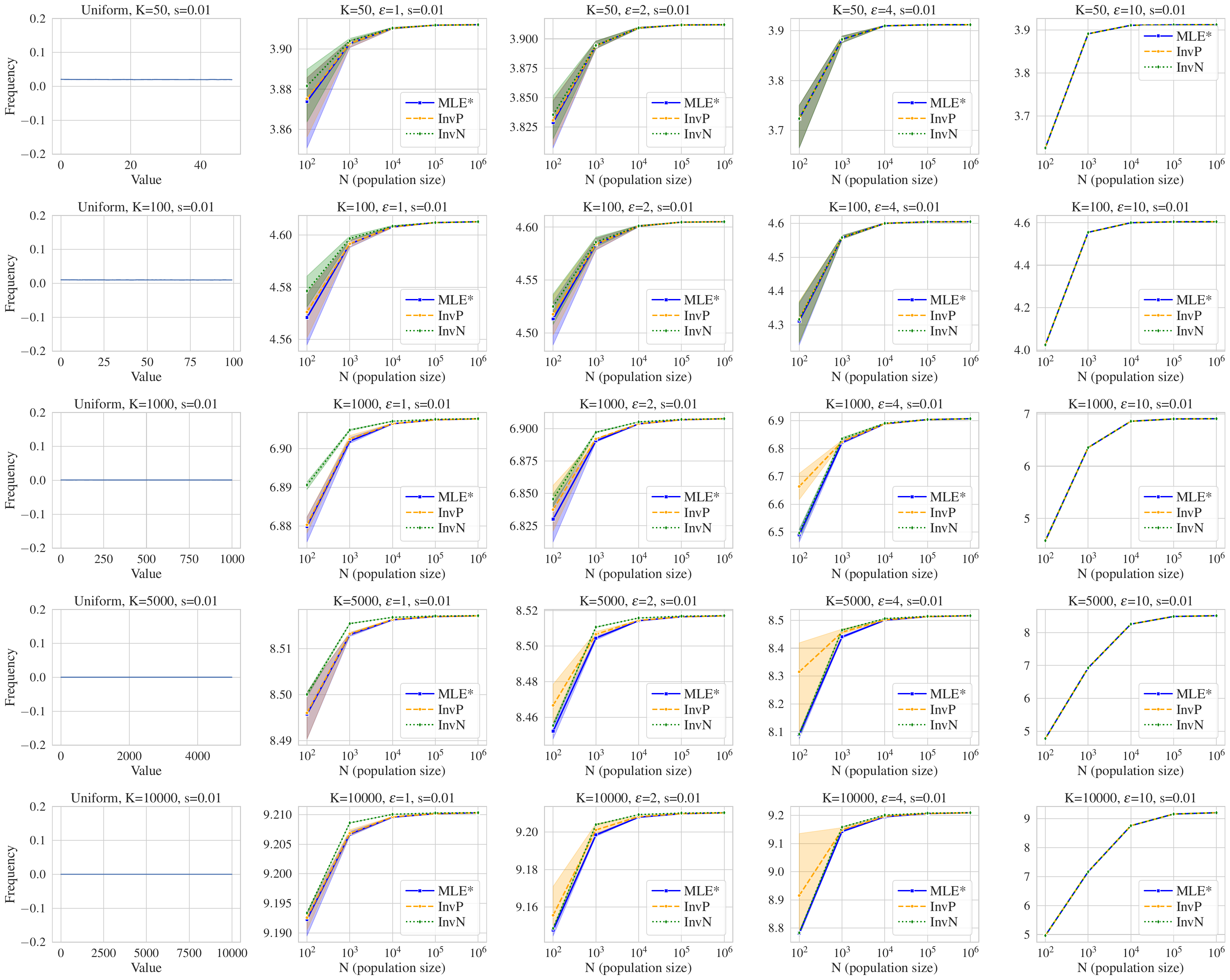}
    \caption{Negative log-likelihood results for near-uniform distribution ($s = 0.01$). 
    Rows vary domain size $K$, and columns vary privacy level $\epsilon$. 
    The far-left subplot of each row shows the true histogram.}
    \label{SM:fig:zipf_nll_s_0.01}
\end{figure*}

\begin{figure*}[!htb]
    \centering
    \includegraphics[width=1\linewidth]{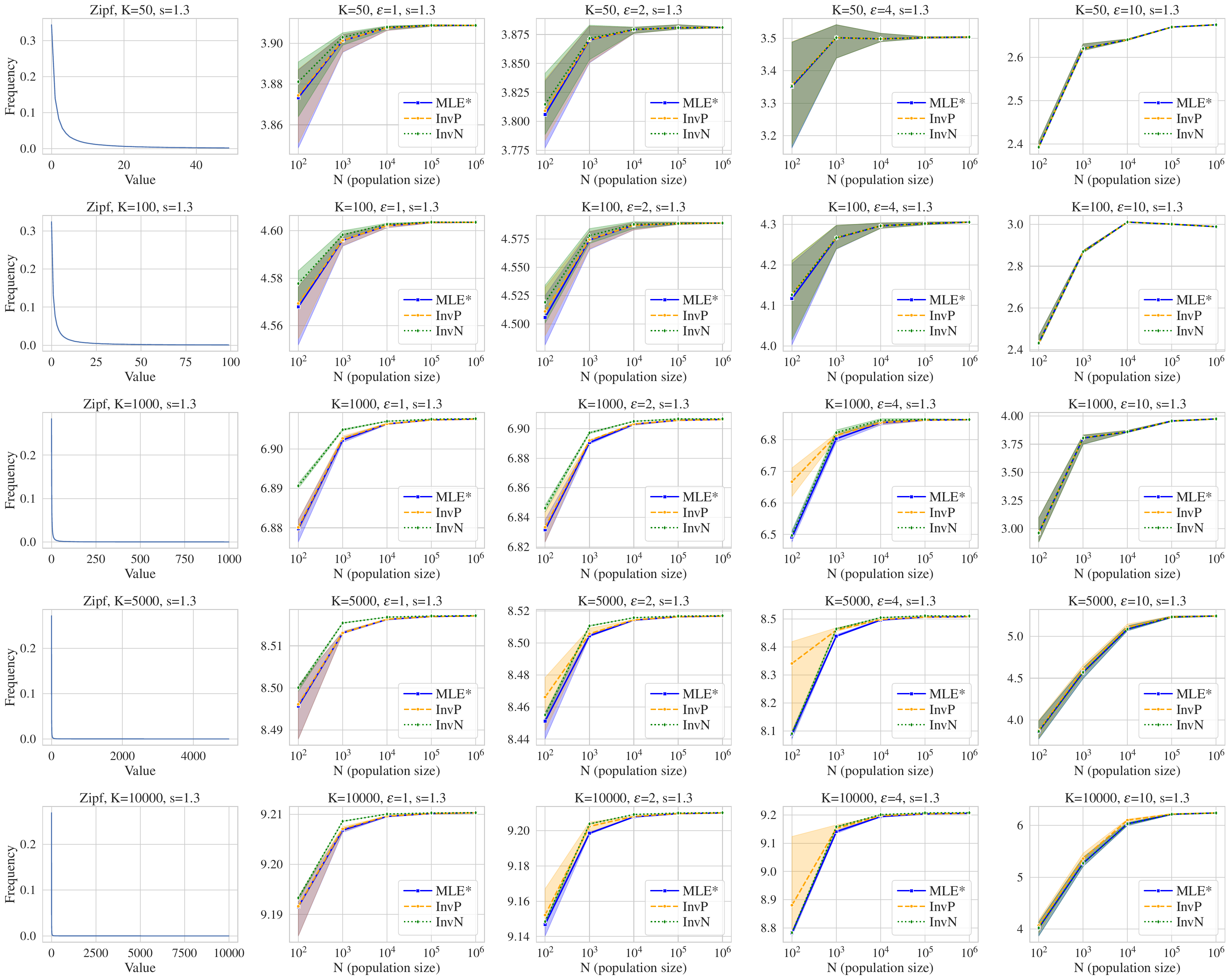}
    \caption{Negative log-likelihood results for moderately skewed distribution ($s = 1.3$). 
    Rows vary domain size $K$, and columns vary privacy level $\epsilon$. 
    The far-left subplot of each row shows the true histogram.}
    \label{SM:fig:zipf_nll_s_1.3}
\end{figure*}

\begin{figure*}[!htb]
    \centering
    \includegraphics[width=1\linewidth]{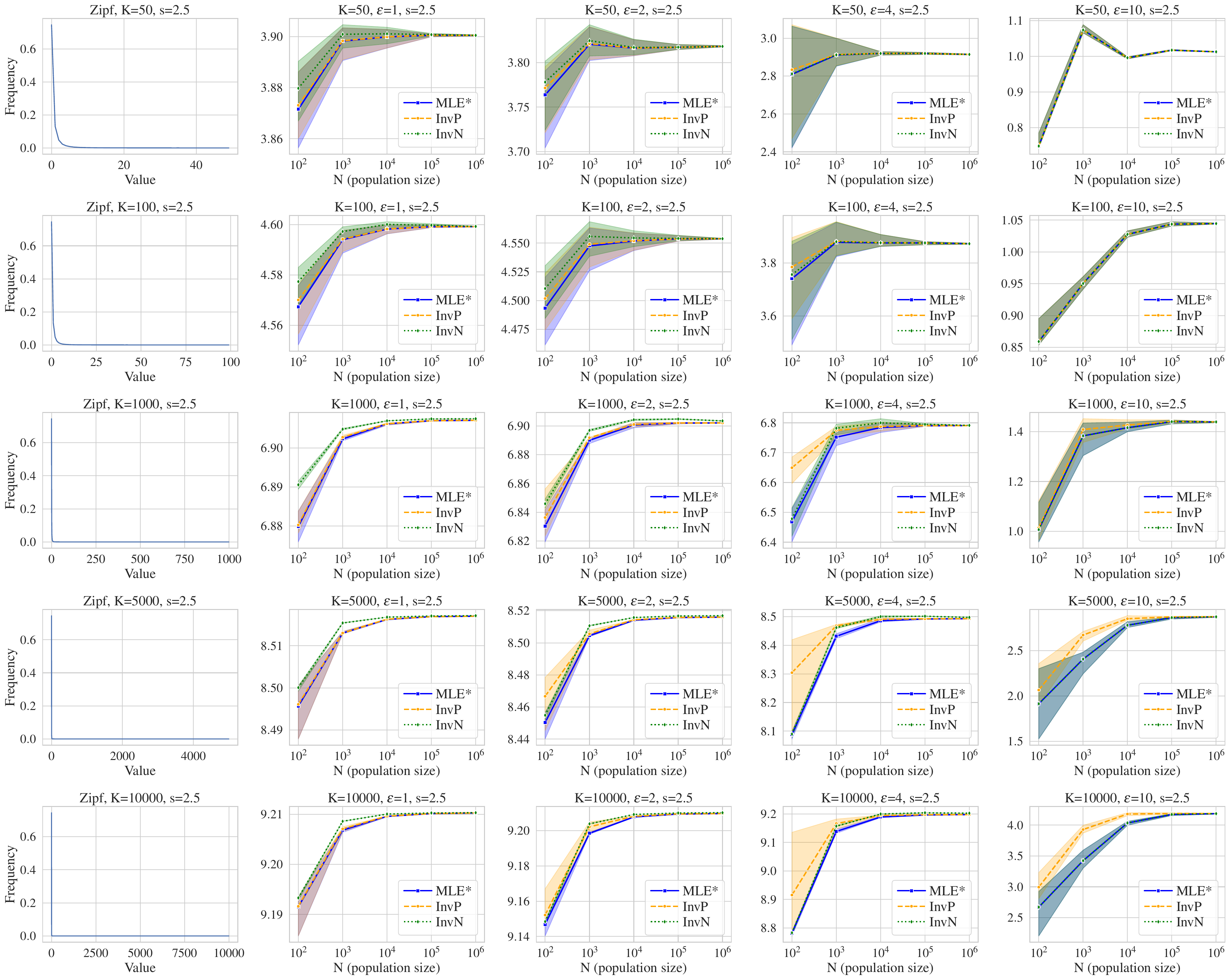}
    \caption{Negative log-likelihood results for highly concentrated distribution ($s = 2.5$). 
    Rows vary domain size $K$, and columns vary privacy level $\epsilon$. 
    The far-left subplot of each row shows the true histogram.}
    \label{SM:fig:zipf_nll_s_2.5}
\end{figure*}